\documentclass[12pt]{article}
\usepackage[margin=0.9in]{geometry}
\usepackage{booktabs}
\usepackage[ruled]{algorithm2e}
\usepackage{amsmath}
\usepackage{amsthm}
\usepackage{caption}
\usepackage{subcaption}
\usepackage{amssymb}

\usepackage{thmtools}
\usepackage{mathtools}
\usepackage{thm-restate}

\usepackage{comment}
\usepackage{xspace}
\usepackage[usenames,dvipsnames]{xcolor}
\usepackage[colorlinks=true,citecolor=ForestGreen,linkcolor=ForestGreen,urlcolor=NavyBlue]{hyperref}
\usepackage[capitalize]{cleveref}

\usepackage{tikz}
\usetikzlibrary{arrows,positioning}

\newtheorem{theorem}{Theorem}[section]
\newtheorem{lemma}[theorem]{Lemma}
\newtheorem{proposition}[theorem]{Proposition}
\newtheorem{definition}[theorem]{Definition}
\newtheorem{corollary}[theorem]{Corollary}
\newtheorem*{remark}{Remark}
\newtheorem{innercustomthm}{Theorem}
\newenvironment{customthm}[1]
  {\renewcommand\theinnercustomthm{#1}\innercustomthm}
  {\endinnercustomthm}
\newtheorem{example}{Example}[section]

\SetAlFnt{\small}
\SetAlCapFnt{\small}
\SetAlCapNameFnt{\small}
\SetAlCapHSkip{0pt}
\IncMargin{-\parindent}

\usepackage{natbib}

\usepackage{footnote}
\makesavenoteenv{algorithm}

\definecolor{english}{rgb}{0.0, 0.5, 0.0}
\definecolor{lilac}{RGB}{230,215,255}
\definecolor{salmon}{RGB}{255,55,40}
\definecolor{customblue}{RGB}{53,140,219}

\newcommand{\x}{\mathbf{x}}

\setcitestyle{authoryear}

\renewenvironment{abstract}
 {\small
  \begin{center}
  \bfseries \abstractname\vspace{-.5em}\vspace{0pt}
  \end{center}
  \list{}{
    \listparindent 1.5em
    \setlength{\leftmargin}{8mm}
    \itemindent    \listparindent
    \setlength{\rightmargin}{\leftmargin}
  }
  \item\relax}
 {\endlist}

\author{
  Luca D'Amico-Wong\\
  Harvard University\\
  \texttt{ldamicowong@g.harvard.edu}
  \and 
  Yannai A. Gonczarowski\\
  Harvard University\\
  \texttt{yannai@gonch.name}
  \and 
  Gary Qiurui Ma\\
  Harvard Univeristy\\
  \texttt{qiurui\char`_ma@g.harvard.edu}
  \and
  David C. Parkes \\
  Harvard University \\
  \texttt{parkes@eecs.harvard.edu}
}

\title{Disrupting Bipartite Trading Networks:\\ Matching for Revenue Maximization\thanks{The authors thank Sid Banerjee for helpful comments and discussions. Gonczarowski gratefully acknowledges research support by the National Science Foundation (NSF-BSF grant No.\ 2343922), Harvard FAS Inequality in America Initiative, and Harvard FAS Dean’s Competitive Fund for Promising Scholarship.}}
\date{}

\begin{document}
\maketitle

\begin{abstract}
We model the role of an online platform disrupting a market with unit-demand buyers and unit-supply sellers. Each seller can transact with a subset of the buyers whom she already knows, as well as with any additional buyers to whom she is introduced by the platform. Given these constraints on trade, prices and transactions are induced by a competitive equilibrium. The platform's revenue is proportional to the total price of all trades between platform-introduced buyers and sellers.

In general, we show that the platform's revenue-maximization problem is computationally intractable. We provide structural results for revenue-optimal matchings and isolate special cases in which the platform can efficiently compute them. Furthermore, in a market where the maximum increase in social welfare that the platform can create is $\Delta W$, we prove that the platform can attain revenue $\Omega(\Delta W/\log(\min\{n,m\}))$, where $n$ and $m$ are the numbers of buyers and sellers, respectively. When $\Delta W$ is large compared to welfare without the platform, this gives a polynomial-time algorithm that guarantees a logarithmic approximation of the optimal welfare as revenue. We also show that even when the platform optimizes for revenue, the social welfare is at least an $O(\log(\min\{n,m\}))$-approximation to the optimal welfare. Finally, we prove significantly stronger bounds for revenue and social welfare in homogeneous-goods markets. 
\end{abstract}

\section{Introduction}
\label{sec:intro}

Online platforms, such as Amazon and Uber Eats, have an increasingly important role to play in the modern economy. As discussed in~\citet{WangMELTZP23}, such platforms bring together buyers and sellers, erase the physical and geographic barriers preventing trade, and generate value in the form of new transaction opportunities. The significance of these platforms was perhaps most evident during the COVID-19 pandemic. As entire populations were put under lockdown, consumers flocked to these platforms, and Uber Eats saw a surge in both supply and demand \citep{raj2020covid} while Amazon's stock price nearly doubled between March 13th and late August of 2020. \citep{icsik2021impact,tiananalysis}.

However, the growing market power of these platforms has also been associated with unwanted strategic practices, such as a platform favoring specific merchants over other, less profitable ones to maximize its own revenue. In 2019, the European Union Commission filed an antitrust lawsuit against Amazon, alleging that the latter had directed buyers to third-party sellers who paid it hefty delivery and storage fees, obscuring better deals elsewhere~\citep{veljanovski2022algorithmic}. While platforms play a crucial role in disrupting markets and increasing social welfare, a platform may also benefit by strategically matching buyers and sellers to further its own gain. In response to these concerns, in this paper we present a theoretical framework for studying the platform's revenue-maximization problem and the 
impact of the platform's strategic behavior on social welfare. 

We adapt a model from~\citet{Alon2023platform} to our setting. We consider a market of \emph{unit-demand} buyers and \emph{unit-supply} sellers that is mediated by an online platform. Each buyer 
can transact with a set of sellers whom she already knows, and in addition with any 
sellers that might be introduced to her by the platform.
Given these constraints on trade, prices and transactions are induced by a competitive (Walrasian) equilibrium.
 The platform can choose to introduce a (possibly distinct) set of sellers to each buyer, and its revenue is proportional to the total prices of all trades between agents that it introduces.\footnote{For simplicity, we do not model this revenue as explicitly coming from the seller or the buyer.} The platform thus seeks to match buyers to sellers in a way that maximizes its revenue. We define the \emph{social welfare} as the sum of transacting buyers' valuations\footnote{We assume that sellers have no costs for trading.} and 
the \emph{optimal (social) welfare} as the social welfare when all sellers are introduced to each buyer. 

We study the computational complexity of the platform's revenue-maximization problem. While we show that the problem is generally intractable, we give structural results for revenue-optimal matchings, identifying specific market settings where the platform can maximize revenue in polynomial time. We also analyze the relationship between revenue optimization and social welfare: We lower bound the platform's revenue as a function of the maximum increase in social welfare that the platform can create, and we upper bound the loss in welfare that results from the platform's self-interest. Leveraging these results, we give a polynomial-time algorithm that guarantees a logarithmic approximation of the optimal revenue in markets where the platform can substantially increase social welfare. We also give tighter bounds for both revenue and social welfare in markets with \textit{homogeneous goods}, where each buyer values all sellers' items equally, but with different buyers having potentially different values.

\subsection{Results and Techniques}
\paragraph{Hardness of the Platform's Revenue-Maximization Problem.}

In Section~\ref{sec:hardness}, we study the computational complexity of maximizing the platform's revenue, to which we refer as the \emph{platform's problem}. 
We assume that the market clears according to {\em maximum} competitive equilibrium prices (i.e., the prices in the seller-optimal competitive equilibrium). In the event that multiple competitive allocations exist, we assume that the platform can break ties in its favor (selecting the equilibrium that maximizes its revenue). This reflects a setting
 in which the platform has sufficient market power to set prices and direct trades.
Intuitively, the platform's problem is challenging in that the platform must simultaneously consider the prices of potential trades, their feasibility in a competitive equilibrium, and the externalities they impose on the prices of other trades.

In a market with $n$ buyers and $m$ sellers, denote the set of sellers that buyer $b_i$ originally knows by $N(i)$, the size of this set as buyer $i$'s \emph{degree}, and buyer $i$'s value for seller $j$'s item by $v_{ij}$. For the platform's problem (with input size $\mathrm{poly}(n, m)$), we establish the following complexity result.

\begin{customthm}{1}[Informal Version of Theorems~\ref{thm:nphard} and \ref{thm:vc_apx_proof}]\label{thm:informal_hardness}
    The platform's revenue-maximization problem is NP-hard, even when every buyer has degree at most two ($|N(i)| \leq 2$ for all $i$) and values all \emph{desired} items equally (i.e., $v_{ij} \in \{0, v_i\}$ for all $j$). If we do not restrict to instances where the buyer has degree at most two, then the problem is APX-hard.
\end{customthm}

The proof of Theorem \ref{thm:informal_hardness} leverages techniques previously used in the profit-maximizing envy-free pricing literature, where we reduce from a version of $3$-SAT in \citet{chen2014envy} to prove NP-hardness and reduce from a version of minimum-vertex cover in \citet{guruswami2005profit} to prove APX-hardness. In contrast to the envy-free pricing problem, where there are no constraints on trade and unsold items can command high prices, our setting exhibits two main differences.
First, buyers have limited access to sellers, and access is mediated by the platform. When no buyer knows any seller, the platform can facilitate trades to achieve the welfare-optimal matching, extracting the optimal welfare as revenue. When buyers know some subsets of sellers, the existing trading possibilities both constrain the platform's revenue per trade and introduce externalities on the prices of trades that the platform facilitates. This is different from envy-free pricing where any buyer can envy any other buyers. When adapting the reductions from \citet{guruswami2005profit} and \citet{chen2014envy}, we introduce extra sellers to the market and carefully select the set of buyers and sellers who can trade outside the platform in order to mimic the envy structure present in the original reductions. Second, while competitive equilibrium prices are envy-free in this setting, competitive equilibria additionally require that unsold items have price $0$. In our proofs, we deal with this latter difference by introducing extra ``dummy'' buyers to transact with unsold items. 

\paragraph{Tractable Special Cases of the Platform's Problem.} 
Theorem~\ref{thm:informal_hardness} shows hardness, even when each buyer has degree at most two and values all desired items equally. We show that the restrictions in \cref{thm:informal_hardness} on the degree and values not being even stronger is not simply an artifact of our proof; rather, this
 corresponds to the frontier of intractability. In Section~\ref{sec:poly-special_case}, we introduce two market settings at the boundary of the
 hardness result of Theorem~\ref{thm:informal_hardness}  where the revenue-optimal matching can be found in polynomial time. 
The first market setting tightens the degree restriction to be at most $1$ and additionally assumes goods are homogeneous, that is, $v_{ij}=v_i$ for all $j$.
The second setting restricts the class of valuations to markets with homogeneous goods and buyers, that is, all buyers value all items equally at some constant positive amount (e.g., $v_{ij}=1$ for all $i,j$); the second setting also places a further restriction on the graph structure, which is already satisfied by the reduction of \cref{thm:informal_hardness}. These two cases demonstrate that starting from the setting of Theorem \ref{thm:informal_hardness}, imposing even slightly more structure on the underlying network and/or the valuations is enough to make the platform's problem tractable.

\begin{customthm}{2}[Informal Version of Theorems~\ref{thm: poly_AMOS} and \ref{thm:poly_identity}]\label{thm:informal_amos_iden}
    The platform's revenue-maximization problem can be solved in polynomial time (in $n$ and $m$) in the following two markets: (i) Homogeneous-goods markets with buyer degree at most $1$; (ii) Markets with homogeneous goods and buyers, in which each buyer has degree at most $2$ and there are sparse connections, so that for any pair of sellers, at most one buyer knows them both.
\end{customthm}

In the homogeneous-goods setting, we give a structural characterization result that shows that the revenue-optimal transactions introduced by the platform connect groups of buyers who know the same seller into \emph{chains} and \emph{cycles} of bounded length. Ranking all buyer groups by the largest buyer value in each group, we show that buyer groups that form a cycle are adjacent in rank. This characterization leads to a  dynamic programming algorithm to pair buyer groups into chains and cycles for maximum revenue, which is at the heart of the proof of Part (i) of \cref{thm:informal_amos_iden}. In markets with homogeneous goods and buyers, we show that finding the revenue-optimal matching is related to finding the maximum set of buyers with non-positive \textit{graph surplus}. In graph theory, the surplus of a vertex set is used in finding maximum matchings \citep{lovasz2009matching}, and is defined as the size of the neighborhood of the vertex set minus the size of the vertex set itself. We prove Part (ii) of \cref{thm:informal_amos_iden} by showing that this buyer set can be found in polynomial time under the structural assumptions of markets in the second special case. 

\paragraph{Bounds between Revenue Optimization and Social Welfare.} After studying the computational aspects of the platform's problem, we shift our focus in Sections~\ref{sec:general_markets} and~\ref{sec:hom_markets} to understanding the relationship between revenue optimization and welfare. On the one hand, when the platform has the potential to increase social welfare substantially, we ask whether such improvements come with a corresponding guarantee for the platform's revenue.
 Conversely, we investigate the impact on social welfare when the platform acts in its own self-interest and optimizes for revenue.
The theorems for each of the two directions are as follows:

\begin{customthm}{3}[Informal Version of Theorem~\ref{thm:gen_welfare_conversion}]
\label{thm:informal_gen_welfare_conversion}
    In any market, if there exists a set of transactions that the platform can add that increases social welfare by $\Delta W$, then the platform can add some subset of this set to generate $\Omega(\Delta W/\log(\min\{n,m\})$ in revenue. This subset can be found from the original set in polynomial time.\footnote{Note that the set of platform edges that maximizes the (increase in) social welfare can be found in polynomial time.} Furthermore, this bound is tight: 
There exists a market where the platform can increase welfare by $\Delta W$ and optimal revenue is $O(\Delta W/\log(\min\{n,m\})$. 
\end{customthm}

\begin{customthm}{4}[Informal Version of Theorem~\ref{thm:poa_upper_bound}]\label{thm:inform_poa}
   In any market, the social welfare is at least an $\Omega(1/\log(\min\{n,m\}))$-fraction of the optimal welfare when the platform maximizes its revenue.
\end{customthm}

Theorem~\ref{thm:informal_gen_welfare_conversion} illustrates a method for the platform to translate the potential welfare increase into revenue. In markets with a substantial welfare gap ($\Delta W$ is large), this
further suggests a polynomial-time approach to guaranteeing as revenue a logarithmic fraction of optimal welfare and,
 consequently, of
 optimal revenue. Combined with Theorem~\ref{thm:informal_hardness}, this result suggests that although the exact revenue maximization problem is hard in general markets, one can nevertheless extract a substantial fraction of optimal revenue in a computationally efficient manner. From another perspective, by giving a revenue guarantee in markets with high inefficiencies (i.e., large $\Delta W$), this theorem can be seen as establishing a motivation for revenue-interested platforms to disrupt such markets, providing a possible explanation for their real-life tendency to do so. In addition, and with an eye toward markets where efficiency is most lacking, we are less interested in settings where $\Delta W$ is small, even when
 optimal revenue can significantly exceed $\Delta W$. Theorem~\ref{thm:inform_poa} discusses the other side of the relationship between welfare and revenue: for all markets, revenue maximization comes with modest welfare guarantees --- namely, revenue maximization cannot harm overall welfare by more than a logarithmic factor.

\paragraph{High Revenue and Social Welfare in Homogeneous-Goods Markets.} 
We also explore the special case of homogeneous-goods markets; recall that these are markets where each buyer values all sellers' items equally but with different buyers potentially having different values. In these markets, optimal welfare is achieved when the buyers with the largest values transact. 
We show in Section~\ref{sec:hom_markets} 
that platform self-interest in this setting is perfectly
aligned with welfare maximization;
i.e., whenever the platform maximizes its own revenue, the optimal welfare is attained. 
Additionally, the platform can obtain revenue at least as large as the welfare gap $\Delta W$ in polynomial time. This implies a polynomial-time procedure offering a constant-factor approximation for revenue when $\Delta W$ is large.

\begin{customthm}{5}[Informal Version of Theorems~\ref{thm:hom_conversion} and \ref{thm:hom_poa}]
\label{custom_thm:hom}
   In homogeneous-goods markets, if there exists a set of transactions that the platform can add that increase welfare by $\Delta W$, there is a polynomial-time algorithm to guarantee platform's revenue of at least $\Delta W$. Furthermore, if the platform's revenue is optimal, then social welfare equals the optimal welfare. 
\end{customthm}

The proof of Theorem \ref{custom_thm:hom} utilizes the structural notion of \textit{opportunity paths}~\citep{kranton2000competition}
 to characterize the competitive price of a trade in homogeneous-goods markets. We show that in such markets, the buyers with the highest values must trade in the revenue-optimal matching, implying that the revenue-optimal matching attains the optimal welfare. We also show that the platform can match buyers to sellers according to the optimal welfare matching while obtaining the full increase in welfare ($\Delta W$)
 as revenue. Unlike in \cref{thm:informal_gen_welfare_conversion} for the general setting, this construction does not restrict itself to a subset of these transactions, offering both optimal welfare and good revenue guarantees.

Overall, our work intertwines computational complexity results with economic findings for the platform's revenue-maximization problem. Despite the computational hurdles that present a barrier to exact maximization in general settings, we develop ways to extract a substantial fraction of the optimal revenue in polynomial time in markets with a large welfare gap, and we further show that social welfare is reasonably well-aligned with revenue maximization, providing welfare guarantees for the revenue-optimal matching in fully general markets.

\subsection{Related Work}
Two works share particularly close models and research questions to our own. \citet{banerjee2017segmenting} model a platform controlling visibility between buyers and sellers of different types in a bipartite market. They study how the platform controls the visibility in order to maximize revenue and welfare. In a similar setting but with full visibility, \citet{birge2021optimal} ask how the platform optimally chooses to set transaction fees. The primary difference of both of these papers from our work is that buyers and sellers in these papers are assumed to have zero utility (equivalently, cannot trade) outside the platform. In contrast, we emphasize the existence of these outside trading opportunities, which play a crucial role both in our modeling and the complexity of the problem. These outside opportunities allow us to model markets that exist prior to platform disruption, rather than focusing only on markets that are initiated by the platform. Outside opportunities also mediate the complexity of the platform's revenue maximization. Without them, the platform's problem can solved in polynomial time in our setting; as more outside trading opportunities are introduced, the externalities they impose significantly complicate the question. We remark that \citet{banerjee2017segmenting} and \citet{birge2021optimal} are more general in other aspects; notably, each node there does not represents a single buyer or seller as in our paper but instead a continuum of buyers or sellers, each represented by a particular demand or supply curve.

Another closely related work is that of \citet{Alon2023platform}, who similarly model a platform that joins a market with existing  trading opportunities. Their focus is on the platform strategically setting a transaction fee, whereas we analyze the platform's matching strategy.

Besides these three papers, our work aligns with the literature that studies the formation of economic markets modeled by networks. \citet{kranton2001theory,kranton2000competition} and \citet{elliott2015inefficiencies} examine a unit-demand, unit-supply market where buyers and sellers strategically invest in costly opportunities to trade. They analyze the effect of the strategic behavior on social welfare. For infinitely divisible goods, \citet{even2007network} analyze all possible outcomes of the market structure when agents strategize to invest in trading opportunities, and \citet{kakade2004economic} examine a Fisher market on a graph. Like these studies, we build on the foundational research on competitive equilibria \citep{kelso1982job, gul1999walrasian}. However, complementing these works where market participants form trading links in a decentralized manner, we focus on a centralized, revenue-maximizing platform that facilitates trades.

Our work relates to the literature at large studying a platform's role in a market. Platforms have been modeled as intermediaries stocking goods, setting prices, and reselling to downstream market participants on a fixed network \citep[e.g.][]{blume2007trading,condorelli2017bilateral,manea2018intermediation,kotowski2019trading}; as guiding search and discovery processes of different market segments \citep[e.g.][]{immorlica2021designing,hu2022dynamic,halaburda2018competing,huttenlocher2023matching}; and as strategic entities that maximize revenue through pricing or matching tools \citep{ke2022information,WangMELTZP23}. Our work differs crucially in the ways in which prices are set; in these works, all prices are set by the platform, whereas we assume that prices are beyond the platform's control, instead induced by market forces modeled by a competitive equilibrium.

For our computational complexity hardness results, we use similar methods to those used in the {\em envy-free pricing} literature \citep{guruswami2005profit,chen2014envy}. There, the goal is to find a revenue-maximizing set of prices (and corresponding allocation) such that each agent receives their most desired good. While we draw on similar techniques, our problem is distinct in that the platform does not directly set prices and is limited in its ability to influence the prices due to buyers' and sellers' existing connections without the platform. This makes it challenging to translate approximation algorithms from this literature to our setting.

\section{Model \& Preliminaries}

Following \citet{Alon2023platform} and \citet{kranton2000competition}, we model a buyer-seller network as a bipartite graph $G = (B, S, E)$. Later, we significantly deviate from these models by introducing a centralized platform that directs trades between buyers and sellers. The network is composed of a set of $n$ unit-demand buyers $B = \{b_1, \dots, b_n\}$ and $m$ unit-supply sellers $S = \{s_1, \dots, s_m\}$. The set of edges $E$ represents the feasible transaction opportunities available to buyers and sellers – that is, buyer $b_i$ and seller $s_j$ can transact if and only if $(b_i, s_j) \in E$. Each buyer $b_i$ has a valuation $v_{ij} \geq 0$ for seller $s_j$'s good.

Given a graph $G$ and valuation profile $\textbf{v}$, the market clears according to a competitive equilibrium, defined by the tuple $(\textbf{x}, \textbf{p})$, where $\textbf{x} \in \{0, 1\}^{n \times m}$ represents the allocation and $\textbf{p} \in \mathbb{R}^m$ corresponds to the item prices (where $\textbf{p} \geq 0$). Formally, $x_{ij} = 1$ if and only if buyer $b_i$ receives seller $s_j$'s item, for which they pay price $p_j$.

\begin{definition}[Competitive Equilibrium on a Buyer-Seller Network]
    Given a graph $G = (B, S, E)$ and valuations $\textbf{v}$, a competitive equilibrium is defined by an allocation, price pair $(\textbf{x}, \textbf{p})$ satisfying:
    \begin{itemize}
        \item Transactions are feasible: $x_{ij} = 0$ for all $(b_i, s_j) \not \in E$.
        \item All items are allocated at most once: $\sum_i x_{ij} \leq 1$ for $1 \leq j \leq m$.
        \item Each buyer receives at most one good: $\sum_j x_{ij} \leq 1$ for $1 \leq i \leq n$.\footnote{While unit-demand buyers could receive more than one good in a Walrasian equilibrium, it is WLOG to assume that they receive only their favorite item in the bundle.}
        \item Each buyer receives their favorite item: $\sum_j x_{ij} \cdot (v_{ij} - p_j) = \max(0, \max_j v_{ij} - p_j)$.
        \item All unsold items have price zero: $\sum_i x_{ij} = 0 \implies p_j = 0$.
    \end{itemize}
\end{definition}

Let $\mathrm{Eq}(G,\mathbf{v})$ denote the set of all competitive equilibria for a graph $G$ and valuations $\mathbf{v}$. \citet{kelso1982job} proved that this set is non-empty when buyers have unit-demand valuations. Furthermore, the First Welfare Theorem states that any competitive equilibrium allocation $\textbf{x}$ maximizes social welfare; in other words, $\textbf{x}$ corresponds to a maximum weight matching on the graph $G$, where edge weights are given by valuations.

\begin{theorem}[First Welfare Theorem]
    If $(\textbf{x}, \textbf{p})$ is a competitive equilibrium, then $\textbf{x}$ maximizes welfare, respecting the constraints on transactions given by edges.
    \begin{align*}
        \sum_{i,j} x_{ij} \cdot v_{ij} = W(G) 
    \end{align*}
    where $W(G)$ denotes the value of the maximum weight matching on $G$.
\end{theorem}

For a market $G=(B,S,E)$ with valuations $\textbf{v}$, we use $W(G) \coloneqq W(B,S,\textbf{v},E)$ to denote the 
{\em social welfare} of this market. This corresponds to the value of the maximum weight matching on $G$. Letting $C_{B,S}$ denote the edges of the \textit{complete bipartite graph}, the \textit{optimal welfare} is denoted as $W^\star(G) \coloneqq W(B, S,\textbf{v},C_{B,S})$. This represents the social welfare when all buyers and sellers can freely transact. When it is clear from the context, we omit $G$ and directly use $W^\star$ to denote optimal welfare.

While the First Welfare Theorem characterizes the set of allocations belonging to a competitive equilibrium, it says nothing about the item prices. In general, there are many sets of prices that could form a competitive equilibrium, with these price vectors forming a lattice \citep{gul1999walrasian}. Throughout this paper, we use the maximum competitive prices, which have a natural interpretation as the seller's contribution to the social welfare.
\begin{theorem}[\citet{gul1999walrasian}]
    Seller $s_j$'s maximum competitive price is given by
    \begin{align*}
        p_j = W(G) - W(G \setminus \{s_j\})
    \end{align*}
    where $G \setminus \{s_j\}$ represents the graph $G$ after removing seller $s_j$ and all incident edges.
\end{theorem}

\subsection{Introducing the Platform}

Now, we model the platform's role in our buyer-seller network. Let $G_w = (B, S, E_w)$ denote the \textit{world graph} without the platform. The \textit{world edges} $E_w$ represent existing transaction opportunities between buyers and sellers; i.e., these buyers and sellers can already transact without the platform, perhaps in-person or through some other non-platform related channels. The platform possesses \textit{full information} about the world graph and the valuation profile $\textbf{v}$. It then chooses a set of \textit{platform edges} $E_p$ to add between buyers and sellers who are not already connected by world edges ($E_p \cap E_w = \emptyset$). We refer to the graph with platform edges as the \textit{platform graph}: $G_p = G_w\cup E_p = (B, S, E_w \cup E_p)$.

With these new edges in the platform graph, the market clears according to a competitive equilibrium with maximum competitive prices. The platform charges a fixed percentage, which we assume to be exogenous, of the total price of all trades on the platform edges.
 We will use $\mathrm{Rev}(E_p)$ to denote the sum of the prices associated with trades that are completed on
 platform edges $E_p$ and $\mathrm{Rev}^\star$ to denote the maximum possible revenue. In reality, the platform's revenue is a fixed percentage of this  quantity. For simplicity, we  refer to the sum of prices along transacting platform edges as platform ``revenue,'' noting that maximizing the two quantities is equivalent.
 
To reflect the platform's market power, we assume it can break ties between multiple equilibria, in addition to selecting for the maximum competitive equilibrium prices. The goal of the platform is to add a set of platform edges (and potentially break ties between competitive equilibria) to maximize its revenue.
\begin{definition}[The Platform's Problem]
    Given a world graph $G_w=(B,S,E_w)$ and valuation profile $\textbf{v}$, the platform seeks to introduce a set of platform edges $E_p$ that maximizes revenue. The optimal revenue is given by
    \begin{align*}
        \mathrm{Rev}^\star = \max_{E_p} \max_{(\textbf{x}, \textbf{p}) \in \mathrm{Eq}(G_p, \textbf{v})} \sum_{(b_i,s_j)\in E_p} p_j \cdot x_{ij}.
    \end{align*}
\end{definition}

\begin{example}
  To illustrate the complexity of the platform's problem, consider the simple market in Figure~\ref{fig:all_platform_edge_bad} where there are no world edges and buyers have valuations $v_{1,1}=1,v_{i,i-1}= v_{i, i} = i$ for $i=2,\ldots n$, where $v_{ij}$ is buyer $b_i$'s valuation for seller $s_j$'s item. The naive strategy of adding all possible platform edges performs very poorly, obtaining revenue $n$ as every transaction occurs at a price of $W(G_p) - W(G_p\setminus\{s_i\}) = 1$. The revenue optimal matching would draw a single platform edge connecting each $b_i$ to $s_i$, obtaining revenue $n(n+1)/2$. The platform must strike a balance between facilitating as many transactions as possible and not adding too many edges and creating unwanted competition, driving transaction prices down.
\end{example}

\begin{figure}[t]
    \centering
    \begin{tikzpicture}[scale=1.3]
        \foreach \i/\label in {1/$b_1$, 2.5/$b_2$, 4/$b_3$, 6/$b_n$}
            \node[draw, shape=rectangle, minimum size=0.6cm] (\label) at (\i, 2) {\label};
            
        \foreach \i/\label in {1/$s_1$, 2.5/$s_2$, 4/$s_3$, 6/$s_n$}
            \node[draw, shape=circle, minimum size=0.6cm] (\label) at (\i, 0) {\label};

        \foreach \x/\y/\w/\pos/\loc in {$b_1$/$s_1$/1/midway/right, $b_2$/$s_1$/2/midway/below right,$b_2$/$s_2$/2/midway/right, $b_3$/$s_2$/3/midway/below right, $b_3$/$s_3$/3/midway/right, $b_n$/$s_n$/n/midway/right}
            \draw[line width=1.5pt, customblue, dashed] (\x) -- node[\pos, \loc, font=\footnotesize] {\w} (\y);
     
        \node at (5, 2) {$\ldots$};
        \node at (5, 0) {$\ldots$};

        \node[left] at (0.5, 2)  {Buyers};
        \node[left] at (0.5, 0)  {Sellers};
    
    \end{tikzpicture}
        \caption{A market where adding all platform edges is arbitrarily bad for revenue. There are no world edges. All nonzero valuations are indicated by dashed blue lines. Buyer $b_1$ has value $1$ for seller $s_1$, Buyer $b_i$ has value $i$ for sellers $s_{i-1}$ and $s_{i}$ for $i=2,\ldots,n$. All other valuations are zero.} \label{fig:all_platform_edge_bad}
\end{figure}
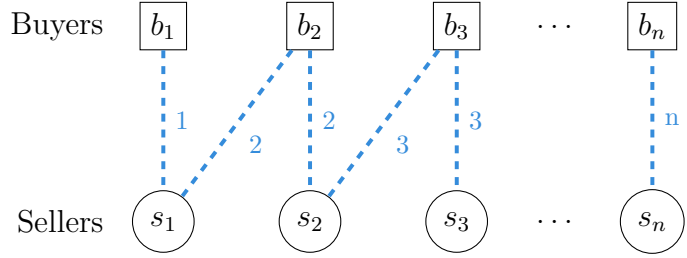 

\subsection{Competitive Prices in Homogeneous-Goods Markets}
To facilitate proofs and more easily identify competitive equilibrium prices in homogeneous-goods markets, we borrow insights from \citet{kranton2000competition}. They introduce the concept of \textit{opportunity paths} to delineate a buyer's direct and indirect competitors and to calculate a buyer's forgone trading opportunity.

\begin{definition}[Opportunity Path \citep{kranton2000competition}]\label{def:oppo_path}
    For a market $G=(B,S,E)$ and an allocation $\textbf{x}$, a buyer $b_i$ is connected to another buyer $b_j$ through opportunity path $$ b_i \mbox{ --- } s_1 \to b_1 \mbox{ --- } s_2 \ldots\ldots b_{j-1} \mbox{ --- } s_j \to b_j$$ where an undirected edge $b_p \mbox{ --- } s_q$ means buyer $b_p$ can trade with seller $s_q$ through existing edges but does not: $(b_p,s_q)\in E, x_{pq}=0$; and a directed edge $s_q \to b_p$ means seller $s_q$ sells to buyer $b_p$: $x_{pq}=1.$
\end{definition}
As an example, consider the market in Figure~\ref{fig:all_platform_edge_bad} when all platform edges are added. Buyer $b_i$ buys from seller $s_i$ and knows $s_{i-1}$. Thus, $b_n$ has opportunity path $b_n \mbox{ --- } s_{n-1} \to b_{n-1} \ldots b_2 \mbox{ --- } s_1 \to b_1$, alternating between transacting and non-transacting edges. Just like alternating or augmenting paths in graph theory, an opportunity path alternates between active edges where transactions occur, and inactive edges where transactions could take place but do not. The ``opportunity" refers to a buyer's outside trading options and uniquely determines the price of an transaction as follows.  
\begin{theorem}[Opportunity Path Theorem \citep{kranton2000competition}]\label{thm:oppo_path}
    For a maximum price competitive equilibrium $(\textbf{x}, \textbf{p})$ in a homogeneous-goods market $G=(B,S,E)$ where $x_{ij}=1$, seller $s_j$'s price is equal to the lowest valuation of any buyer connected to $b_i$ through an opportunity path. If $b_i$ is connected to a seller that does not sell through an opportunity path, $p_j = 0$.
\end{theorem}

\section{Hardness of the Platform's Problem}
\label{sec:hardness}

We begin by showing the platform's problem is computationally intractable, even under relatively significant restrictions on both the underlying structure of the world graph and the valuation structure of the buyers.

\begin{theorem}
\label{thm:nphard}
    The decision version of the platform's problem is NP-hard, even when buyers have degree at most two and value all desired goods equally ($v_{ij} \in \{0, v_i\}$).
\end{theorem}

Similar to the hardness proof for revenue-maximizing envy-free pricing in \citet{chen2014envy}, 
Theorem~\ref{thm:nphard} is the result of a reduction from a variant of 3-SAT and Lemma~\ref{lem:reduction_proof}. In contrast to their proof, where an item's price can be high despite the item being unsold, we design our reduction with world edges to provide buyers with outside trading opportunities and hence limit the price they pay. We also introduce ``dummy'' buyers to absorb any items that are not sold – ensuring that these unsold items do not introduce unwanted externalities in the form of low prices. 

Clauses and variables in the 3-CNF are associated with buyers, and literals are associated with sellers. Based
on the value of literals satisfying the assignment, prices that buyers with world edges to corresponding literal sellers pay change, reflected in the optimal revenue; the reduction is designed such that a valid assignment to the 3-CNF exists if and only if the optimal revenue exceeds a given threshold. 

\paragraph{Reduction from a Variant of 3-SAT}
Consider a modified version of 3-SAT, where each variable $x_i$ appears positively ($x_i$) and negatively $(\bar{x}_i)$ an equal number of times. Given a 3-CNF $\varphi$, we add clauses $d_i = (x_i \vee \bar{x}_i)$ for each variable $x_i$. Additionally, we pad $\varphi$ with these clauses such that each variable $x_i$ appears exactly $2t$ times for some $t > 0$.  Denote this modified 3-CNF by $\varphi' = (c_1 \wedge \dots \wedge c_k)$. Let $q$ be the number of unique variables, and let $Z = q \cdot k \cdot t$. 

Given this 3-CNF $\varphi'$, we construct a corresponding instance of the platform's problem as follows. There are $k + 2(t-1)q + (2tq - k) = (4t-2)q$ buyers in total and an equal number of items. We describe the construction below: 

\begin{itemize}
    \item For each variable $x_i$, we add $2t$ items $\alpha_{i,j}, \tau_{i,j}$ for $j \in \{0, \dots, t-1\}$. The $\alpha_{i,j}$ represent the positive instances of $x_i$ while the $\tau_{i,j}$ represent the negative instances. WLOG, let $\alpha_{i,0}, \tau_{i,0}$ be the items that correspond to clauses $d_i = (x_i \vee \bar{x}_i)$. 
    \item For each clause $c_i$ in $\varphi'$, we add a buyer $U_i$ that has value $Z$ for each item that represents the corresponding literals in clause $c_i$. 
    \item For each variable $x_i$, we add $2(t-1)$ buyers $A_{i,j}, T_{i,j}$, $j \in \{1, \dots, t-1\}$. We also add $2(t-1)$ items $\gamma_{i,j}, \delta_{i,j}$, $j \in \{1, \dots, t-1\}$. Each buyer $A_{i,j}$ has value $N+1$ for $\alpha_{i,0}, \alpha_{i,j}, \gamma_{i,j}$, and each buyer $T_{i,j}$ has value $Z+1$ for $\tau_{i,0}, \tau_{i,j}, \delta_{i,j}$. For each of these buyers, we add world edges to the $\alpha_{i,0},\alpha_{i,j},\tau_{i,0},\tau_{i,j}$ they have positive value for.
    \item Finally, we add $2tq - k$ dummy buyers who have value $H \geq Z+1$ for all $\alpha_{i,j}, \tau_{i,j}$.
\end{itemize}
All buyers have value $0$ for all sellers not mentioned above, and no other world edges are present besides those explicitly mentioned. We prove this is a valid reduction.
\begin{restatable}{lemma}{SATReduction}
\label{lem:reduction_proof}
    There is a valid assignment to the original CNF if and only if the optimal revenue for the platform is at least $D \coloneqq kZ + q(t-1)(2Z+1) + H(2tq - k)$. 
\end{restatable}
The proof argument shares similar logic with that of \citet{chen2014envy} and is given in Appendix~\ref{app:sat_np_proof}. The main difference with our proof arises in proving the ``if" direction. In our construction, buyers $A_{i,j}$ and $T_{i,j}$ always transact with $\gamma_{i,j}$ and $\delta_{i,j}$ through platform edges. At the same time, we use world edges to connect these buyers to $U_i$, forming opportunity paths. Buyer $U_i$ corresponds to the literal $x_i$. In this construction, $A_{i,j}$ and $T_{i,j}$ pay $Z$ or $Z+1$, depending on the true value of the literals in the CNF. This change is reflected in the optimal revenue and helps to establish the ``if" direction.

Theorem~\ref{thm:nphard} shows that even when introducing a small number of existing interdependencies between buyers and sellers through world edges, computing the revenue-optimal matching  becomes intractable. By further relaxing the constraint on the degree of the buyers, we prove the platform's problem is APX-hard. In other words, unless $P = NP$, the platform's problem does not admit a polynomial-time approximation scheme.
\begin{restatable}{theorem}{VertexCoverReduction}
\label{thm:vc_apx_proof}
    The decision version of the platform's problem is APX-hard, even when buyers value all desired goods equally ($v_{ij} \in \{0, v_i\}$). 
\end{restatable}
Like \citet{guruswami2005profit} for maximum revenue envy-free pricing, our proof reduces from minimum vertex cover on graphs with degree at most $K$. For an instance of the vertex cover problem, we construct a market with buyers that correspond to vertices and edges. The edge buyers always buy at price $1$, but a vertex buyer buys at price $1$ or $2$ depending on whether the vertex belongs to the minimum vertex cover. To achieve the different prices, we use world edges to build opportunity paths between the vertex buyers and edge buyers, so that the number of the vertex buyers paying price $1$ is at least the size of the minimum vertex cover. Finally we add extra ``dummy'' buyers to make sure no seller has price zero. The full reduction and proof is presented in Appendix \ref{app:vc_apx_proof}.

\paragraph{Differences from Envy-Free Pricing.} While there are connections between the platform's problem and profit-maximizing envy-free pricing \citep{guruswami2005profit, balcan2008item}, positive results there are less easily portable to our setting for two major reasons. First, the platform does not directly set prices; rather, prices are determined as a function of the edges that the platform chooses to add. Second, envy-free pricing allows unsold items to have high prices, effectively eliminating them from the market, where they do not affect prices of other trades. In contrast, a competitive equilibrium requires that unsold items clear at price zero. We bypass this difference in our proofs of hardness by adding dummy buyers so all sellers transact. However, this presents a major barrier when adapting approximation algorithms from the envy-free pricing literature. Since the platform cannot directly set prices, to remove the effect of a seller $s_j$ on other trades, there needs to be a buyer with high valuation for $s_j$, but if no such buyer exists, then there is nothing the platform can do. In contrast, in envy-free pricing, one could simply set $p_j$ to be high enough such that the item is irrelevant.

\section{Tractable Special Cases}\label{sec:poly-special_case}
While the previous section shows the platform's problem is hard in general, here we demonstrate there are structural properties of the markets that reduce the complexity of the problem. In Section~\ref{sec:structural_prop}, we characterize properties that hold for the platform's optimal matching in general unit-supply, unit-demand markets. In Sections~\ref{sec:special-case-homogeneous} and \ref{sec:special-case-identical}, we present two classes of markets where the platform's problem can be solved in polynomial time. These markets represent slight restrictions of the settings that are NP-hard (Theorem \ref{thm:nphard}), showing that our proof of NP-hardness is in this sense  ``tight.''

\subsection{Structural Properties of Revenue Maximum Matchings}\label{sec:structural_prop}

In spite of the hardness of the platform's problem, we are able to prove some structural properties that hold for the platform's optimal matching strategy more generally. 
Both lemmas below will play a role when finding revenue-optimal matchings in special markets in Section~\ref{sec:special-case-homogeneous} and \ref{sec:special-case-identical}, as well as the proof of hardness in Appendix~\ref{app:vc_apx_proof}.

First, we show that the platform has no incentive to add more than one edge incident on any one buyer/seller. Thus, we can think of the platform as recommending a single transaction to a given buyer or seller; any further recommendations are extraneous.
\begin{restatable}{lemma}{atmostone}
\label{lem:at_most_one_edge}
In any general market, there exists a revenue-optimal matching that adds at most one edge to each buyer and each seller. Furthermore, in this matching, all platform edges transact at positive prices in the competitive equilibrium.
\end{restatable}
\begin{proof}
    Let $E_t$ denote the set of edges where transactions take place in the competitive equilibrium. Assume that the platform adds two or more edges to a buyer or seller in $E_p$, or an edge that does not transact. Since buyers are unit-supply and sellers are unit-demand, there exists an edge $e_{ij}=(b_i, s_j) \in E_p$ such that $e_{ij} \notin E_t$. Removing $e_{ij}$ does not change seller $s_j$'s price, because 
    \begin{eqnarray*}
        W(G_p \setminus \{e_{ij}\}) & = &W(G_p)\\
        W(G_p\setminus\{s_j, e_{ij}\}) & = & W(G_p \setminus \{s_j\})
    \end{eqnarray*}
     However, removing $e_{ij}$ weakly increases other sellers' $s'_{j}$ prices, because 
     \begin{eqnarray*}
         W(G_p\setminus\{s'_j, e_{ij}\})\leq W(G_p \setminus \{s'_j\})
     \end{eqnarray*}
    Thus the platform weakly prefers to drop all edges $e_{ij}$ that are not transacting in the final competitive equilibrium. The same argument applies for transactions with price zero.
\end{proof}

Intuitively, any non-transacting edges that the platform adds only increase the level of competition between different sellers. By removing these edges, sellers are able to charge higher prices, and in turn, the platform's revenue from the transactions it facilitates increases. So when searching for the revenue-optimal matching, the platform has no strict incentives to add platform edges that do not transact.

Next, we show that the platform has an incentive to make sure that as many buyers and sellers transact as possible. This lends further support to the hypothesis that the platform's incentive to maximize revenue can come hand-in-hand with social welfare maximization.
\begin{restatable}{lemma}{allTransact}
\label{lem:all_transact}
    In any general market, there exists a revenue-optimal matching where either all sellers sell, or all buyers buy, or buyers and sellers who do not trade have value zero for each other.
\end{restatable}
\begin{proof}
    We consider the case when there are weakly more buyers than sellers, though the opposite case follows similarly. Take any equilibrium where a seller $s_j$ does not transact. As there are at least as many buyers as sellers, there must be a buyer $b_i$ who also does not transact. For the sake of contradiction, assume that $v_{ij} > 0$. Note that there cannot exist an edge connecting $b_i$ to $s_j$, or else this would violate the First Welfare Theorem. Consider connecting the two via a new platform edge $e_{ij}=(b_i, s_j)$. 

     From this new transaction, we obtain revenue
    \begin{align*}
        W(G_p \cup \{e_{ij}\}) - W(G_p \cup \{e_{ij}\} \setminus \{s_j\}) = W(G_p \cup \{e_{ij}\}) - W(G_p) = v_{ij} > 0
    \end{align*}
    which is strictly positive.
    
    Now, consider any other seller $s'_j$, and let $G_p$ be the graph before adding edge $e_{ij}$. We have that
     \begin{align*}
         p_{s'_j}(G_p) &= W(G_p)-W(G_p \setminus \{s'_j\})\\
         p_{s'_j}(G_p\cup\{e_{ij}\}) &= W(G_p)+v_{ij}-W(G_p\cup\{e_{ij}\} \setminus\{s'_j\})
     \end{align*}
     We will show that $p_{s'_j}(G_p\cup\{e_{ij}\})\geq p_{s'_j}(G_p)$. It suffices to prove that
    \begin{align*}
        v_{ij} \geq W(G_p\cup\{e_{ij}\} \setminus\{s'_j\}) - W(G_p \setminus \{s'_j\})
    \end{align*}
    If the max weight matching in $G_p\cup\{e_{ij}\} \setminus\{s'_j\}$ does not use the new edge $e_{ij}$, then the right hand side is equal to zero and we are done. Thus, suppose that $e_{ij}$ is in the max weight matching. Then
    \begin{align*}
       &W(G_p\cup\{e_{ij}\} \setminus\{s'_j\})=W(G_p \setminus \{b_i, s_j, s'_j\})+v_{ij}\leq W(G_p \setminus \{s'_j\})+v_{ij}
    \end{align*}
    and this is precisely the inequality we want to prove. It follows that any other seller $s'_j$'s price weakly increases, so by matching $s_j$ to $b_i$, the platform's revenue strictly increases.
\end{proof}

\subsection{Single World Seller, Homogeneous-Goods Markets}\label{sec:special-case-homogeneous}

We begin by considering the class of homogeneous-goods markets ($v_{ij} = v_i$),
and in addition requiring that each buyer has at most one incident world edge. We call these markets SWSH (single world seller, homogeneous-goods) markets. This structure allows us to nicely decompose the world graph into what we will call {\em seller subgraphs}, {\em dangling buyers}, and {\em dangling sellers}.

\begin{definition}[Seller Subgraphs]
     In SWSH markets, take any seller $s_j$ who knows at least one buyer via a world edge. The {\em seller subgraph} $S_j$ is the collection of buyers and edges $S_j:=\{s_j\}\cup \{b_i|(b_i,s_j)\in E_w\} \cup \{(b_i, s_j) \in E_w\}$ adjacent to seller $s_j$ via world edges. We will use $v(S_j):=\max_{(b_i,s_j)\in E_w}v_i$ to denote the value of the largest buyer connected to $s_j$.
\end{definition}

There can be multiple buyers in a seller subgraph. All buyers and sellers are either part of a seller subgraph or have no incident world edges. The buyers and sellers with no incident world edges are termed \emph{dangling buyers} and \emph{dangling sellers}. By definition, platform edges cannot be part of any seller subgraphs. In a platform graph $G_p$, platform edges connect seller subgraphs into \emph{cycles} and \emph{chains}. We define these structures and show that they are closely related to the platform's revenue.

\begin{definition}[Cycles]
    For a SWSH market $G=(B,S,E)$ and an allocation $\textbf{x}$, a cycle of $k$ seller subgraphs is defined by 
    $$ S_1 \to S_2 \to S_3 \ldots S_k \to S_1,$$
    where $S_i \to S_j$ means that in allocation $\textbf{x}$, seller $s_i$ sells to a buyer $b_j\in S_j$ in seller subgraph $S_j$.
\end{definition}

\begin{definition}[Chains]
    For a SWSH market $G=(B,S,E)$ and an allocation $\textbf{x}$, a chain of $k$ seller subgraphs is defined by 
    $$ (s_0) \to S_1 \to S_2 \ldots S_k \to (b_k),$$
    where $S_i \to S_j$ means in allocation $\textbf{x}$, where $s_i$ buys from a buyer $b_j\in S_j$ in seller subgraph $S_j$. For the {\em starting subgraph} $S_1$, there can be a dangling seller $s_0$ selling to a buyer in $S_1$. For the {\em terminal subgraph} $S_k$, $s_k$ either does not sell, sells to a buyer $b_k$ who is dangling or $b_k\in S_k$, or belongs to a cycle that does not contain subgraphs $S_1,\ldots S_k$. A chain must include a terminal subgraph.
\end{definition}
We use {\em $k$-cycle} and {\em $k$-chain} to denote cycles and chains consisting of $k\geq 1$ seller subgraphs, or of length $k$. We use {\em $0$-chain} to denote a dangling seller $s_0$ selling to a buyer in a chain or cycle. A $1$-cycle is a seller subgraph $S_i$ where the seller $s_i$ transacts through world edges with a buyer $b_i\in S_i$. All sellers in a cycle transact within the cycle, while one seller $s_k\in S_k$ in the terminal subgraph of a chain can sell outside of the chain. This makes it possible for a chain to be attached to a cycle through $s_k$. 
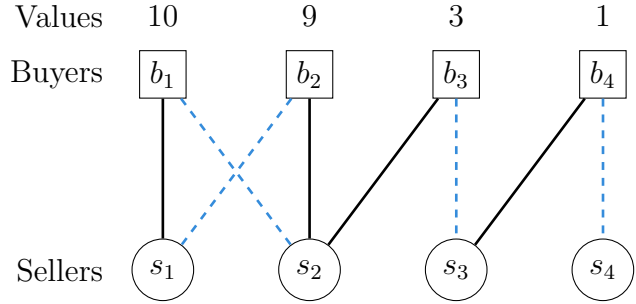
\begin{figure}[hbtp]
    \centering
    \begin{tikzpicture}[scale=1.3]
        \foreach \i/\label in {1/$b_1$, 2.5/$b_2$, 4/$b_3$, 5.5/$b_4$}
            \node[draw, shape=rectangle, minimum size=0.6cm] (\label) at (\i, 2) {\label};
            
        \foreach \i/\label in {1/$s_1$, 2.5/$s_2$, 4/$s_3$, 5.5/$s_4$}
            \node[draw, shape=circle, minimum size=0.6cm] (\label) at (\i, 0) {\label};

        \foreach \x/\y in {$b_1$/$s_2$, $b_2$/$s_1$, $b_3$/$s_3$, $b_4$/$s_4$}
            \draw[line width=1pt, customblue, dashed] (\x) -- (\y);

        \foreach \x/\y in {$b_1$/$s_1$, $b_2$/$s_2$, $b_3$/$s_2$, $b_4$/$s_3$}
            \draw[line width=1pt, black, solid] (\x) -- (\y);

        \node at (1, 2.6)  {10};
        \node at (2.5, 2.6)  {9};
        \node at (4, 2.6)  {3};
        \node at (5.5, 2.6)  {1};

        \node[left] at (0.5, 2.6)  {Values};
        \node[left] at (0.5, 2)  {Buyers};
        \node[left] at (0.5, 0)  {Sellers};
    
    \end{tikzpicture}
        \caption{An example of a SWSH market. World edges are depicted in black, and the optimal set of platform edges is in blue. In revenue optimal matching, $s_1$ sells to $b_2$, $s_2$ sells to $b_1$, $s_3$ sells to $b_3$ and $s_4$ sells to $b_4$. 
        } \label{fig:hom_char_example}
\end{figure}

\begin{example} \label{example:cycle_chain_illustrate}
To illustrate the definitions, consider the example market presented in Figure~\ref{fig:hom_char_example}. The set of optimal platform edges decomposes seller subgraphs into cycles and chains as follows:
    \begin{align*}
       \underbrace{s_4 \; {\to}}_{\text{0-Chain}}\;  \underbrace{S_3\; {\to} }_{\text{$1$-Chain}} \; \underbrace{S_2 \to S_1 \to S_2}_{\text{$2$-Cycle}}
    \end{align*}
\end{example}

The definitions above and the structure of SWSH markets directly yield the following proposition.
\begin{proposition}\label{prop:cycles_chains}
    In SWSH markets, any set of transacting platform edges uniquely connects seller subgraphs into chains and cycles.
\end{proposition}
Lemma~\ref{lem:at_most_one_edge} says there exists a revenue optimal matching where all platform edges transact. By Proposition~\ref{prop:cycles_chains}, these platform edges uniquely define a set of cycles and chains. Thus, finding the revenue-optimal matching is equivalent to finding the chains and cycles with the highest revenue. For ease of illustration, in this section we consider a market where the number of buyers $n$ equals the number of sellers $m$ \footnote{Our proof in Appendix~\ref{app:swsh_reduction} works for all cases of $n$ and $m$.}, and all buyers and sellers transact according to Lemma~\ref{lem:all_transact}. In addition, from here forward, we will suppress any mention of 0-chains and deal only with chains and cycles of positive length, leaving implicit the fact that any buyer who does not transact with sellers as part of a cycle or a chain will be matched to a dangling seller.

Recall that Theorem~\ref{thm:oppo_path} states that the price of a trade in a homogeneous-goods market is determined by the lowest buyer valuation on a buyer's opportunity path. This is precisely what the cycle and chain structure captures. Consider the $0$-chain and the $1$-chain in Example~\ref{example:cycle_chain_illustrate}. Buyer $b_4$'s opportunity path goes through the chain, reaching all buyers. As she has the lowest valuation on this chain, she pays price $1$. Buyer $b_3$ also has the lowest valuation on her own opportunity path, so she pays price $3$. Similar logic applies to the cycle. 
    
However, if $b_2$ and $b_3$ were to swap the sellers that they trade with, both $b_2$ and $b_1$ would have $v_3=3$ as the lowest valuation on their opportunity paths, and they would thus pay $3$ instead of $9$. As one can design the set of cycles and chains to go through only the buyer with the largest valuation in a given seller subgraph, these highest value buyers (i.e., $v(S_j)$ for $S_j$) are crucial for revenue maximization in SWSH markets.

We are now ready to present our results. Assume there are $\ell\leq m$ seller subgraphs in a SWSH market. We first sort all seller subgraphs by their valuations $v(S_1) \geq v(S_2) \geq \cdots \geq v(S_\ell)$. This induces a partial order over all seller subgraphs. We call a $k$-cycle \emph{contiguous} if it comprises of seller subgraph $S_i, S_{i+1}, \ldots S_{i+k-1}$ for some $i$ – that is, no intermediate seller subgraph is skipped. We call a chain \emph{contiguous} if there exists an index $i$ such that the chain is of form $S_\ell \to S_{\ell-1} \to \ldots \to S_{i-1} \to S_i$ and has terminal subgraph $S_i$. Our algorithm to compute a revenue-optimal matching relies on the following structural result, whose proof we give through a series of lemmas in Appendix~\ref{app:swsh_char}.
\begin{restatable}{theorem}{SWSHChar}
\label{thm:opt_char}
    In any SWSH market, there exists a set of revenue-optimal platform edges that, through the largest buyers in each seller subgraph,  connects seller subgraphs into contiguous cycles of length at most $3$, and at most one contiguous chain.
\end{restatable}
 
At a high level, the proof first shows that any cycle of length more than three can be split up into smaller cycles yielding weakly more revenue. Then if a chain or a cycle is not contiguous, we can further divide and/or join them to weakly increase revenue. We now use dynamic programming to leverage this characterization result and efficiently search over the possible optimal configurations as described in Theorem~\ref{thm:opt_char}.
\begin{theorem}\label{thm: poly_AMOS}
    For SWSH markets where $n=m$, there exists a polynomial-time algorithm to find the set of revenue-optimal platform edges.
\end{theorem}

\begin{proof} 
Fix a seller subgraph $S_k$ as the terminal subgraph for the potential chain. $S_k$ uniquely identifies a contiguous chain. There are $\ell\leq m$ such possible chains. Consider the induced graph $G'$ obtained by removing all the subgraphs in the chain from the original graph. 

Let $\mathrm{DP}[i]$ denote the maximum obtainable revenue from connecting the largest $i$ subgraphs in $G'$ via cycles. Let $\mathrm{Rev}(S_{i,j})$ denote the maximum obtainable revenue from the contiguous cycle formed by subgraphs $S_i, S_{i+1},\ldots, S_{j}$. When $j=i$, $\mathrm{Rev}(S_{i,i})=0$. As contiguous cycles have length at most three, $\mathrm{DP}[i]$ can be calculated through dynamic programming. Filling in the base cases for $DP[1], DP[2], DP[3]$, we have the following recurrence for $i=4,...,k-1$:
\begin{align*}
    &DP[i] = \max (DP[i-3] + \mathrm{Rev}(S_{i-2,i}),  DP[i-2] + \mathrm{Rev}(S_{i-1,i}), DP[i-1] + \mathrm{Rev}(S_{i,i}))
\end{align*}
Let $\mathrm{R}[k]$ be the maximum obtainable revenue from the potential chain that is identified by $S_k$. Since this chain is contiguous, $\mathrm{R}[k]$ can be calculated in $O(n^2)$ time, by comparing the revenue when connecting the chain to all possible buyers in $G'$ through $s_k$. It follows that the optimal revenue when there is a contiguous chain identified by $S_k$ is given by $DP[k-1] + R[k]$. Repeating this process for all $\ell$ possible chains, the optimal revenue is given by $\max_{k=1}^{\ell} DP[k-1] + R[k]$. Finding the overall maximum takes $O(n^3)$ time.

Note the above procedure exhaustively searches all configurations of contiguous cycles of length at most 3 and at most one contiguous chain. The procedure thus finds the optimal configuration in terms of revenue. By Theorem~\ref{thm:opt_char}, there exists an optimal solution that satisfies these properties.
\end{proof}

While Theorem~\ref{thm: poly_AMOS} gives a polynomial-time algorithm for the case in which $n = m$, we show that the more general case reduces to this case by carefully discarding some buyers and sellers from the graph. We present the proof of Theorem~\ref{thm:swsh_reduction} in Appendix~\ref{app:swsh_reduction}.

\begin{restatable}{theorem}{SWSHReduction}
\label{thm:swsh_reduction}
    There exists a polynomial-time algorithm to maximize platform's revenue in general SWSH markets.
\end{restatable}
\begin{proof}[Proof Sketch]
    When $n > m$, we show that we can discard all buyers besides those with the $m$ highest valuations without affecting optimality. Similarly, when $m > n$, we show that we can get rid of $m - n$ dangling sellers without affecting optimal revenue. Thus, we are left with the case in which $m = n$, for which Theorem \ref{thm: poly_AMOS} gives us the desired result.
\end{proof}

\subsection{Sparse Homogeneous-Goods-and-Buyers Markets}\label{sec:special-case-identical}
While Theorem~\ref{thm:nphard} only states hardness for the case where buyers have degree at most two, one can check that the reduction yields a world graph that satisfies an additional property:
    {\em For any pair of sellers, at most one buyer knows them both.}
We call the world graphs that satisfy this additional property ``sparse graphs''. By considering these sparse graphs with buyers of degree at most 2, and restricting the class of valuations from $v_{ij} \in \{0, v_i\}$ to $v_{ij} = c$ for some constant $c>0$, we show that the platform's problem becomes tractable. We call these markets SHGB (Sparse Homogeneous-Goods-and-Buyers) markets. We begin by showing that finding the optimal set of platform edges in SHGB markets reduces to a graph-theoretic problem involving sets of buyers with non-positive {\em surplus}. Proofs for this section are intricate and are all presented in Appendix~\ref{app:hv_identity}. 
\begin{definition}[Surplus]
    In a bipartite market $G=(B,S,E_w)$, consider a set of buyers $B_v$. Let $N(B_v)$ be the set of all sellers who are adjacent to at least one buyer in $B_v$ via a world edge. The surplus of $B_v$ is defined as $|N(B_v)| - |B_v|$\footnote{For a set $B_v$, we use $|B_v|$ to denote its cardinality.}.
\end{definition}
\begin{restatable}{lemma}{IdentityChar}
\label{lemma:identity_char}
    In SHGB markets, there exists a set of platform edges generating revenue $x$ if and only if there exists a set of buyers $B_v$ of non-positive surplus such that
    $$\min\{|B_v|, |S|\} - |N(B_v)| + k_v = x,$$
    where $k_v$ is the cardinality of the maximum matching between $B_v$ and $N(B_v)$ using platform edges.
\end{restatable}

For a set of buyers $B_v$ with non-positive surplus, the proof minimizes the number of buyers that connect to an unsold item through opportunity paths. Revenue $k_v$ is generated from buyers that trade with sellers in $N(B_v)$, and revenue $\min\{|B_v|,|S|\} - |N(B_v)|$ is attained from the remaining buyers who do not transact with $N(B_v)$.
We then prove that except for a few cases that can be checked in polynomial time, $k_v=|N(B_v)|$. This further simplifies the platform's problem. 

\begin{restatable}{lemma}{MaxCard}
\label{lemma:max_cardinality}
    In SHGB markets, the platform can find the set of revenue-optimal platform edges by finding the maximum set of buyers with non-positive surplus.
\end{restatable}

With this reduction in hand, we give a polynomial-time algorithm to find the maximum set of buyers with non-positive surplus. To do this, we rely on the graph structure in SHGB markets, identifying the buyers and sellers in these markets with edges and vertices of a general graph. Our problem then translates to finding the set of edges in this graph whose induced subgraph contains no more than a certain number of vertices. This is achievable using standard graph algorithms.
\begin{restatable}{theorem}{PolyIdentity}
\label{thm:poly_identity}
    There exists a polynomial-time algorithm to solve the platform's problem
 in SHGB markets.
\end{restatable}

\section{Guarantees for General Markets}\label{sec:general_markets}
Having established the hardness of the platform's problem in the general case and shown some special cases that can be solved in polynomial time, we  turn to analyzing the relationship between revenue optimization and social welfare.

In Section~\ref{sec:welfare_into_revenue}, we provide a lower bound for revenue based on the potential welfare increase the platform can generate. This shows when the platform can effectively disrupt the market and improve overall welfare, the platform's revenue is guaranteed to be large. Specifically, we show in Section~\ref{sec:log_rev_guarantee} that there is a polynomial-time algorithm to extract a logarithmic approximation of the optimal welfare as revenue in such markets. In Section~\ref{sec:PRM}, we study the other direction of this relationship. We prove that in markets where the platform optimizes for revenue, social welfare is at least an $\Omega(1/\log(\min\{n,m\}))$-fraction of the optimal welfare. These results explain the sense in which revenue maximization can go hand-in-hand with welfare considerations in a platform economy.

\subsection{Converting Potential Welfare Increase to Revenue}\label{sec:welfare_into_revenue}

There are two primary ways that the platform can make revenue. The first way is by ``monopolizing'' the existing welfare of the world graph. 

\begin{example}\label{exam:mono_welfare}
    To illustrate this ``monopolization,'' consider a graph with two buyers and a single seller. Buyer $b_1$ has value $1$ for the item and is connected to the seller via a world edge while buyer $b_2$ has value $1 + \epsilon$ and has no incident world edges. In this case, the platform can add a single platform edge from $b_2$ to the seller. Now, buyer $b_2$ transacts instead of buyer $b_1$, and the platform obtains revenue $1 + \epsilon$ while providing only $\epsilon$ in added welfare.
\end{example}

The second way is by facilitating new high-value transactions that were not possible without the assistance of the platform. In these cases, the platform has the potential to add a substantial amount of welfare to the world graph and claim some of the value for itself. The question is how much of this increase in welfare the platform is able to retain as revenue.
 We give a polynomial-time process to convert the potential increase in welfare into revenue below.
\begin{theorem}
\label{thm:gen_welfare_conversion}
    In a general market, suppose there exists a set of $k$ platform edges that increase social welfare by $\Delta W=W(G_p)-W(G_w)$. Then one can find a subset of these edges that yields revenue at least $\Delta W/H_k$ in polynomial time, where $H_k=\sum_{i=1}^k 1/i$ is the $k$th harmonic number.
\end{theorem}
\begin{proof}
    We prove this inductively. When $k = 1$, there is a single platform edge $(b_1, s_1)$ that can be added to increase welfare by $W$. As removing $s_1$ is weakly worse than removing edge $(b_1, s_1)$, the revenue/price of this edge is given by $W(G_p) - W(G_p \setminus \{s_1\})\geq W(G_p) - W(G_p \setminus \{(b_1,s_1)\})\geq \Delta W$.

    Now, suppose that it holds for $k = n - 1$. Consider a set of $k = n$ platform edges that adds welfare $W$. If every edge yields revenue at least $\frac{\Delta W}{n \cdot H_n}$, then it follows that we are guaranteed revenue at least $n \cdot \frac{\Delta W}{n \cdot H_n} = \frac{\Delta W}{H_n}$ and we are done. If not, then we can remove an edge $(b_i,s_j)$ that yields revenue strictly less than $\frac{\Delta W}{n \cdot H_n}$. By removing this edge, we can only decrease the added welfare by at most $W(G_p)-W(G_p\setminus\{(b_i,s_j)\})\leq W(G_p)-W(G_p\setminus\{s_j\}) \leq   \frac{\Delta W}{n \cdot H_n}$. Thus, we are left with a set of $n - 1$ platform edges that adds welfare at least
    $$
        \Delta W - \frac{\Delta W}{n \cdot H_n} = \frac{n H_n - 1}{n \cdot H_n} \cdot \Delta W = \frac{H_n - \frac{1}{n}}{H_n} \cdot \Delta W = \frac{H_{n-1}}{H_n} \cdot \Delta W.
    $$
    Applying our inductive hypothesis, we can guarantee revenue at least
    $$
        \frac{1}{H_{n-1}} \cdot \frac{H_{n-1}}{H_n} \cdot \Delta W = \frac{\Delta W}{H_n}
    $$
    from some subset of these $n-1$ platform edges, concluding the proof.
\end{proof}

The proof argument in Theorem~\ref{thm:gen_welfare_conversion} yields the following simple greedy algorithm.
 Given a set of $k$ platform edges that increase social welfare by $\Delta W$, iteratively remove the edge that offers the least revenue to the platform until left with a single edge. 
Keep track of the revenue at each step, and select the set of edges that yields the highest revenue overall. The proof argument
 shows that this set will guarantee revenue 
 at least $\Delta W/H_k$, and this process  takes
 polynomial time.

In order to achieve $\Delta W/H_k$ as revenue, the platform adds a subset of all edges that contribute to the $\Delta W$ welfare increase, not necessarily all of the edges. Proposition~\ref{prop:tight_welfare_conversion} further shows the $H_k$ conversion rate from the potential welfare increase to revenue is tight.
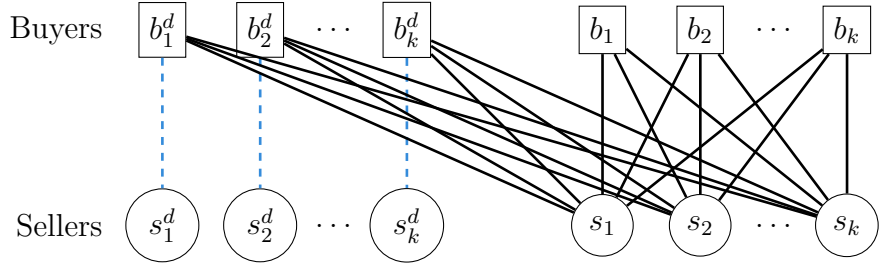
\begin{figure}[t]
    \centering
    \begin{tikzpicture}[scale=1.3]
        \foreach \i/\label in {1/$b_1^d$, 2/$b_2^d$, 3.5/$b_k^d$}
            \node[draw, shape=rectangle, minimum size=0.6cm] (\label) at (\i, 2) {\label};

        \foreach \i/\label in {5.5/$b_1$, 6.5/$b_2$, 8/$b_k$}
            \node[draw, shape=rectangle, minimum size=0.6cm] (\label) at (\i, 2) {\label};
            
        \foreach \i/\label in {1/$s_1^d$, 2/$s_2^d$, 3.5/$s_k^d$}
            \node[draw, shape=circle, minimum size=0.6cm] (\label) at (\i, 0) {\label};

        \foreach \i/\label in {5.5/$s_1$, 6.5/$s_2$, 8/$s_k$}
            \node[draw, shape=circle, minimum size=0.6cm] (\label) at (\i, 0) {\label};

        \foreach \x/\y in {$b_1^d$/$s_1^d$, $b_2^d$/$s_2^d$, $b_k^d$/$s_k^d$}
            \draw[line width=1pt, customblue, dashed] (\x) -- (\y);

       \foreach \x/\y in {$b_1^d$/$s_1$, $b_1^d$/$s_2$, $b_1^d$/$s_k$, $b_2^d$/$s_1$, $b_2^d$/$s_2$, $b_2^d$/$s_k$, $b_k^d$/$s_1$, $b_k^d$/$s_2$, $b_k^d$/$s_k$, $b_1$/$s_1$, $b_1$/$s_2$, $b_1$/$s_k$, $b_2$/$s_1$, $b_2$/$s_2$, $b_2$/$s_k$, $b_k$/$s_1$, $b_k$/$s_2$, $b_k$/$s_k$}
            \draw[line width=1pt, black, solid] (\x) -- (\y);
     
        \node at (2.75, 2) {$\ldots$};
        \node at (7.25, 2) {$\ldots$};
        \node at (2.75, 0) {$\ldots$};
        \node at (7.25, 0) {$\ldots$};

        \node[left] at (0.5, 2)  {Buyers};
        \node[left] at (0.5, 0)  {Sellers};
    
    \end{tikzpicture}
        \caption{The market used in the proof of Proposition~\ref{prop:tight_welfare_conversion}. Buyers $B_k=\{b_1,b_2,...,b_k\}$ and dummy buyers $B^d_k=\{b^d_1,b^d_2,...,b^d_k\}$ are fully connected to sellers $S_k=\{s_1, s_2,...,s_k\}$ through world edges, denoted by solid black edges. Dummy sellers $S^d_k=\{s^d_1, s^d_2,...,s^d_k\}$ are not connected to any buyers through world edges.
        Buyer $b_i\in B_k$ values all sellers in $S_k$ at $1/i$. Buyer $b^d_i\in B^d_k$ values all sellers at $1$. All other valuations are zero. The maximum welfare the platform can add through platform edges is $\Delta W = H_k$, indicated in dashed blue edges. The optimal revenue is $1$, obtained by adding any non-empty subset of the blue edges.} \label{fig:conv_tight}
\end{figure}

\begin{restatable}{proposition}{tightWelfareConversion}
\label{prop:tight_welfare_conversion}
     For all $k$, there exists a world graph $G_w$ and a set of $k$ edges adding welfare $\Delta W$ such that the optimal platform revenue is precisely $\Delta W/H_k$.
\end{restatable}
\begin{proof}
    For any $k$, Figure~\ref{fig:conv_tight} presents such a market.
 World edges are solid in black and platform edges are dashed in blue. There are buyers $B_k=\{b_1,b_2,...,b_k\}$, dummy buyers $B^d_k=\{b^d_1,b^d_2,...,b^d_k\}$, sellers $S_k=\{s_1, s_2,...,s_k\}$ and dummy sellers $S^d_k=\{s^d_1, s^d_2,...,s^d_k\}$. All buyers $B_k$ and $B^d_k$ are fully connected to sellers $S_k$ through world edges. Dummy sellers $S^d_k$ have no connections in the world graph. Buyer $b_i\in B_k$ values all sellers in $S_k$ at $1/i$. Buyer $b^d_i\in B^d_k$ values all sellers at $1$. All other valuations are zero. By adding platform edges to form a perfect matching between the dummy buyers and the dummy sellers, the platform can add welfare $\Delta W = \sum_{i=1}^k 1/i = H_k$ to the world graph. However,   the platform's optimal revenue is bounded above by $1$. The platform has no incentive to add platform edges to $B_k$ as they have value $0$ for all dummy sellers. Thus, the only thing the platform can do is match $\ell\in \{1,\ldots,k\}$ dummy buyers to the dummy sellers. If the platform matches $\ell$ of the dummy buyers to the dummy sellers, each transaction occurs at price $1/\ell$, and total revenue is $\ell \cdot \frac{1}{\ell} = 1$. Thus, we 
 have that the optimal revenue is given by $\Delta W/H_k = H_k/H_k = 1$.
\end{proof}

\subsection{Logarithmic Revenue Guarantees in Inefficient Markets}\label{sec:log_rev_guarantee}

Theorem~\ref{thm:gen_welfare_conversion} shows that the platform can convert a potential welfare increase into revenue in polynomial time. Consider now a minimal set of platform edges such that optimal welfare $W^\star$ is achieved in the resulting platform graph. If we apply Theorem~\ref{thm:gen_welfare_conversion} to this set of platform edges, then we can lower bound the platform's revenue as a function of the maximum increase in welfare.
\begin{corollary}
\label{corr:gen_welf_conv}
    In a general market where the platform's maximum contribution to welfare is $W^\star-W(G_w)$, the platform can obtain revenue $\frac{W^\star - W(G_w)}{H_{\min(n, m)}}$ in polynomial time.
\end{corollary}

If the optimal welfare $W^\star$ is large relative to welfare in the world graph $W(G_w)$, then the platform can add a substantial amount of welfare. In this case Corollary~\ref{corr:gen_welf_conv} implies that the platform can also obtain a substantial amount of revenue. Theorem \ref{thm:log_approx} makes this  precise.
\begin{theorem}
\label{thm:log_approx}
    In general markets where $W^\star - W(G_w) = \Omega(W(G_w))$, the platform can obtain $\frac{W^\star}{O(\log(\min(n, m)))}$ in revenue. This implies an $O(\log(\min(n, m)))$ polynomial-time approximation algorithm for revenue.
\end{theorem}
\begin{proof}
    We have $W^\star = (W^\star - W(G_w)) + W(G_w) = \Theta(W^\star - W(G_w))$. The rest of the proof follows from Corollary~\ref{corr:gen_welf_conv} and the fact that $\mathrm{Rev}^\star \leq W^\star$.
\end{proof}

\begin{remark}
  Theorem~\ref{thm:log_approx} does not just imply a logarithmic approximation to the optimal revenue in inefficient markets. It shows a stronger result: the platform can efficiently obtain a logarithmic fraction of the optimal welfare as revenue in this class of markets. 
\end{remark}

Though Theorem~\ref{thm:log_approx} only provides a revenue guarantee for markets with large potential welfare improvements, these are the markets that we care about and the ones platforms typically aim to enter and disrupt. This theorem can also be seen as a motivation for why revenue-interested platforms tend to disrupt markets with large inefficiencies. In markets that are already relatively efficient, optimal revenue can indeed be much larger than the maximum welfare increase, as shown in Example~\ref{exam:mono_welfare}, but sellers and buyers may have few incentives to use the platform in the first place.

\subsection{Bounding the Welfare Loss from Revenue Maximization}\label{sec:PRM}

Having shown how to lower bound revenue with optimal welfare, we now explore the converse question: how much welfare is guaranteed under the revenue-optimal matching? In other words, we want to understand the welfare loss, compared to the optimal welfare, due to the platform inefficiently matching buyers and sellers in order to maximize its own revenue.
 To formalize this loss in welfare, we define the {\em price of revenue maximization}.
\begin{definition}[Price of Revenue Maximization (PRM)]
    For any market $G_w=(B,S,E_w)$ and valuation profile $\textbf{v}$, let $\Pi_p^\star(G_w, \textbf{v})$ denote the set of all revenue-optimal platform edge configurations. Each element $E^\star_p\in \Pi_p^\star(G_w, \textbf{v})$ is a configuration of revenue-optimal platform edges in market $G_w$.
    The {\em price of revenue maximization} is the largest ratio of the optimal welfare to the welfare of revenue-maximizing platform graph, across any market and valuation profile:
    \begin{align*}
        \mathrm{PRM} = \max_{G_w, \textbf{v}, E^\star_p \in \Pi_p^\star(G_w, \textbf{v})} \frac{W^\star(G_w)}{W(B,S,E_w\cup E^\star_p)}.
    \end{align*}
\end{definition}

To better understand the extent to which the platform's incentives can negatively affect social welfare, we present the following bounds on the PRM for general markets.
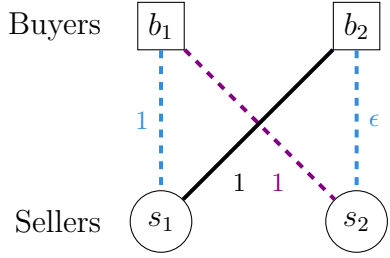
\begin{figure} 
    \centering
    \begin{tikzpicture}[scale=1.3]
        \foreach \i/\label in {1/$b_1$, 3/$b_2$}
            \node[draw, shape=rectangle, minimum size=0.6cm] (\label) at (\i, 2) {\label};
            
        \foreach \i/\label in {1/$s_1$, 3/$s_2$}
            \node[draw, shape=circle, minimum size=0.6cm] (\label) at (\i, 0) {\label};

        \foreach \x/\y in {$b_2$/$s_1$}
            \draw[line width=1pt, blue, dashed] (\x) -- (\y);

        \foreach \x/\y/\w/\pos/\loc in {$b_1$/$s_1$/1/midway/left, $b_2$/$s_2$/$\epsilon$/midway/right}
            \draw[line width=1.5pt, customblue, dashed] (\x) -- node[\pos, \loc, font=\footnotesize] {\w} (\y);

        \foreach \x/\y/\w/\pos/\loc in {$b_1$/$s_2$/1/near end/below left}
            \draw[line width=1.5pt, violet, dashed] (\x) -- node[\pos, \loc, font=\footnotesize] {\w} (\y);

        \foreach \x/\y/\w/\pos/\loc in {$b_2$/$s_1$/1/near end/below right}
            \draw[line width=1.5pt, black, solid] (\x) -- node[\pos, \loc, font=\footnotesize] {\w} (\y);

        \node[left] at (0.5, 2)  {Buyers};
        \node[left] at (0.5, 0)  {Sellers};
    
    \end{tikzpicture}
        \caption{A two-buyer, two-seller market in which the $\mathrm{PRM}$ approaches $2$. The world edge is marked as a black solid edge. Buyer values are annotated next to edges. The welfare-optimal platform edge is in purple while the revenue-optimal set of platform edges is shown in blue. \label{fig:poa_lower_bound}}
\end{figure} 

\begin{restatable}{proposition}{poaLowerBound}
\label{prop:poa_lower_bound}
    The price of revenue maximization is bounded below by $2$.
\end{restatable}
\begin{proof}
    Consider the  two-seller, two-buyer market in Figure~\ref{fig:poa_lower_bound}. The welfare-optimal matching adds edge $(b_1, s_2)$ for a total welfare of $2$ and revenue of $1$. The revenue-optimal matching adds edges $(b_1, s_1), (b_2, s_2)$ for a total welfare and revenue of $1 + \epsilon$. This yields a ratio of $\frac{2}{1 + \epsilon} \to 2$ as $\epsilon \to 0$. Here, we again see the two primary ways in which the platform can obtain revenue. In the revenue-optimal strategy, the platform forgoes adding the socially optimal edge $(b_1, s_2)$ in favor of "monopolizing" the revenue from seller $s_1$ by adding edge $(b_1, s_1)$ instead. This allows it to facilitate another transaction, namely $(b_2, s_2)$, that would not be possibly under the welfare-optimal configuration.
\end{proof} 

While Proposition~\ref{prop:poa_lower_bound} demonstrates that the platform's strategic behavior may come at the cost of social welfare, the procedure of Theorem~\ref{thm:gen_welfare_conversion} allows us to bound this cost from above. Namely, the difference in optimal welfare and the welfare of revenue-optimal matching cannot exceed a logarithmic factor; otherwise, we could convert this difference in welfare into strictly more revenue than we obtain from the revenue-optimal matching.
\begin{theorem}
\label{thm:poa_upper_bound}
    In general markets, the PRM is bounded above by $H_{\min(n, m)} + 1 = O(\log(\min(n, m)))$.
\end{theorem}
\begin{proof}
    We show $PRM \leq H_{\min(n, m)} + 1$. Suppose otherwise, and there exists a market $G_w$ that satisfies $W^\star > (H_{\min(n, m)} + 1) \cdot W(G_w \cup E_p^\star)$, where $E_p^\star$ is a set of revenue-optimal platform edges. Then, there would exist a set of $k\leq \min(n, m)$ edges $E'_P$ such that, by Theorem \ref{thm:gen_welfare_conversion}, we have
    \begin{align*}
        \mathrm{Rev}(E'_P) &\geq \frac{W^\star - W(G_w)}{H_{k}} \geq \frac{W^\star - W(G_w \cup E_p^\star)}{H_{\min(n, m)}} > \frac{(H_{\min(n, m)} + 1) W(G_w \cup E_p^\star) - W(G_w \cup E_p^\star)}{H_{\min(n, m)}}\\
        &= W(G_w \cup E_p^\star) \geq \mathrm{Rev}(E_p^\star),
    \end{align*}
    and this is a contradiction as $E_p^\star$ is revenue-optimal.
\end{proof}

\section{Guarantees for Homogeneous-Goods Markets}\label{sec:hom_markets}

While the results of Section \ref{sec:general_markets} offer  guarantees on both welfare loss and approximately optimal revenue, it is natural to consider if other markets might admit stronger still
 guarantees. As seen in the proof of Proposition \ref{prop:tight_welfare_conversion}, even when we turn to the case in which each buyer  values all desired items equally $(v_{ij} \in \{v_i, 0\})$, the bound on welfare conversion into revenue is already tight.
 This motivates us to consider the next most natural class of valuations $(v_{ij} = v_i)$.
 In this section, we develop very strong welfare and revenue guarantees in homogeneous-goods markets.

\subsection{Converting Potential Welfare Increase to Revenue without Loss}

We first consider whether the platform can more efficiently convert the potential welfare increase in homogeneous-goods markets into revenue. Theorem \ref{thm:hom_conversion} shows that the platform can achieve this conversion 
perfectly, without any loss.
\begin{theorem}
\label{thm:hom_conversion}
    In homogeneous-goods markets, the platform can always extract
 revenue at least $W^\star - W(G_w)$ in polynomial time.
\end{theorem}
\begin{proof}
    As goods are homogeneous, the optimal welfare is obtained by always matching the buyers with the top $\min(m, n)$ highest valuations. Fix some arbitrary maximum weight matching $M^\star$ in the complete graph, and let $M$ be the max weight matching in $G_w$. Let $B_{old} = \{\tilde{b}_1, \dots, \tilde{b}_{\ell}\}$ represent the set of buyers matched in $M$ but unmatched in $M^\star$, and let $B_{new} = \{b^\star_1, \dots, b^\star_k\}$ represent the set of buyers that are matched in $M^\star$ but are unmatched in $M$. Weakly more buyers are matched in $M^\star$ than in $M$; otherwise the leftover sellers could always match with the buyers in $B_{old}$ that they were matched to in $M$. Thus, we have $k \geq \ell$. We  show that the platform can add edges and extract full welfare as revenue from each buyer in $B_{new}$, implying a lower bound of revenue $W^\star - W(G_w)$.

    Let $v_{(\min(n, m))}$ denote the $\min(n, m)$-th largest valuation among all buyers. Since in $M^\star$ the largest $\min(n, m)$ buyers transact, it must be the case that every buyer in $B_{new}$ has valuation weakly greater than $v_{(\min(n, m))}$ and every buyer in $B_{old}$ has valuation strictly less than $v_{(\min(n, m))}$. For each buyer $b^\star_i \in B_{new}$, add a platform edge to the seller that $\tilde{b}_i$ transacted with in $M$, or to a seller who does not sell in $M$ if $i > \ell$. Since $k-\ell$ more buyers transact in $M^\star$ compared to $M$, there are at least $k-\ell$ sellers who do not sell in $M$. Denote the set of platform edges added as $E_p$. Then all buyers in $B_{new}$ still transact in the maximum matching of $G\cup E_p$ through $E_p$. We now prove that every buyer $b_i^\star\in B_{new}$ pays the full price $v_{i}^\star$ by reasoning about the buyers' opportunity paths.

    Since $b^\star_i$ does not transact in $M$, there are no sellers with strictly smaller values on any of its opportunity paths. Otherwise, transactions would change along the opportunity path and $b^\star_i$ would transacted in $M$. As buyers in $B_{old}$ have valuation strictly less than $v_{\min(n,m)}$, they are not on $b^\star_i$'s opportunity path. Thus, in the max weight matching of $G_w\cup E_p$, where $B_{old}$ no longer transacts, $b^\star_i$ remains the buyer with smallest valuation on its own opportunity path and pays $v^\star_i$.
\end{proof}

We make two observations about the procedure in Theorem~\ref{thm:hom_conversion} to extract $W^\star-W(G_w)$ in revenue.
First, the procedure  also maximizes social welfare. Second, to extract revenue $W^\star-W(G_w)$, the platform needs to carefully add platform edges. There are platform edges that increase total welfare by $W^\star-W(G_w)$ 
that do not yield the desired revenue. As an example, consider the market  in Figure~\ref{fig:conv_tight}
and used in the proof of Proposition~\ref{prop:tight_welfare_conversion}. If we modify the market such that all buyers have homogeneous valuations, then adding the dashed blue edges, while yielding the optimal welfare, still only gives revenue $1$. To extract the total amount of added welfare, $H_k$, the platform must add edges between $b_i$ and $s_i^d$. 

As in the case of general markets, Theorem~\ref{thm:hom_conversion} naturally leads to an approximation algorithm for optimal welfare (and hence optimal revenue) when the platform has a way to add
 a large amount of welfare. In these market settings, we can guarantee a constant-factor approximation to optimal revenue, improving on the logarithmic guarantee given for general valuations.
\begin{theorem}
\label{thm:hom_const_approx}
    In homogeneous-goods markets where $W^\star - W(G_w) = \Omega(W(G_w))$, there exists a polynomial-time algorithm that yields a constant-factor approximation of the optimal welfare as revenue to the platform. This implies a constant-factor approximation to the optimal revenue.
\end{theorem}
\begin{proof}
    By Theorem \ref{thm:hom_conversion}, we can always guarantee $W^\star - W(G_w)$ in revenue in polynomial time. Since $W^\star - W(G_w) = \Omega(W(G_w))$, we have that $W^\star - W(G_w) = \Omega(W^\star)$. The final part of the theorem follows from the fact that $\mathrm{Rev}^\star \leq W^\star$.
\end{proof}

While we choose to focus on the case where $W^\star - W(G)$ is large, as we believe these are the markets platforms are most likely to enter, we show that we can more generally extract a $1/\min\{n,m\}$-fraction of the optimal revenue in polynomial time in any homogeneous-goods market. We defer the proof to Appendix~\ref{app:min_nm_approx_rev}.
\begin{restatable}{theorem}{HomoMinNM}
\label{thm:hom_min_nm_approx}
    In homogeneous-goods markets, there exists a polynomial-time algorithm that yields a $\min\{n,m\}$-approximation of the optimal revenue.
\end{restatable}
In order to guarantee a $\min\{n,m\}$-approximation of the optimal revenue, we show that the platform can efficiently maximize the revenue it obtains from a single transaction. As the platform can facilitate at most $\min\{n,m\}$ transactions in total, this yields the desired approximation ratio.

\subsection{Zero Welfare Loss from Revenue Maximization}
As with general markets, we now turn to bounding the welfare loss incurred as a result of the platform's self-interest. Adopting the same approach as  used in proving Theorem~\ref{thm:poa_upper_bound} for general markets, Theorem~\ref{thm:hom_conversion} implies that the price of revenue maximization is bounded above by $2$. While this is already a strong result, we can in fact  show that the price of revenue maximization in homogeneous-goods markets is  $1$. That is, revenue maximization  aligns exactly with welfare maximization in homogeneous-goods markets. 
\begin{restatable}{theorem}{PRMhomo}
\label{thm:hom_poa}
    In homogeneous-goods markets, the price of revenue maximization equals $1$.
\end{restatable}

\begin{proof}
    Let $\tilde{B}$ represent the set of buyers with the top $\min(n, m)$ highest valuations – i.e. the buyers who transact in the welfare-optimal matching. Suppose for the sake of contradiction that the revenue-optimal matching is not welfare optimal and thus some buyer $b_i \in \tilde{B}$ with $v_i \neq 0$ does not transact in the revenue-optimal matching. All opportunity paths from $b_i$ must lead to buyers with weakly higher values. If not, transactions can change along the opportunity path, resulting in a matching with strictly larger weight. This contradicts with the first welfare theorem of competitive equilibrium. 

    If there is a non-transacting seller $s_j$ in the revenue-optimal matching, then adding the edge $(b_i, s_j)$ strictly increases revenue. All buyers with opportunity paths previously ending at $s_j$ paid a price zero as seller $s_j$ did not transact; after the addition of this edge, they pay price $v_i > 0$, leading to a strict increase in revenue.

    If all sellers transact in the revenue-optimal matching, let $b_k$ be the transacting buyer with the smallest valuation and suppose that they currently transact with seller $s_\ell$. As the revenue-optimal matching is assumed not to be welfare optimal, we must have that $v_k < v_i$. By adding the edge $(b_i, s_\ell)$, revenue again strictly increases. All buyers with opportunity paths previously leading to $s_\ell$ must pay a weakly higher price now, since $v_i > v_k$ and all opportunity paths from $b_i$ lead to buyers with weakly higher values. In addition, we obtain revenue $v_i$ from the edge $(b_i, s_\ell)$ and lose revenue at most $v_k < v_i$ from the exclusion of edge $(b_k, s_{\ell})$, in the event that this was a platform edge. Thus, this leads to a strict increase in revenue.

    In either case, there exists a modification to the revenue-optimal matching that yields strictly higher revenue, concluding the proof.
\end{proof}

\section{Discussion}

We have studied the incentives facing platforms when  facilitating transactions between buyers and sellers, modeled by a bipartite buyer-seller network to which the platform can choose to add new links, introducing buyers to sellers.
The general problem of maximizing the platform's revenue is computationally hard, even under restrictive structural assumptions. By imposing additional structure, we give polynomial-time algorithms for  special classes of markets, establishing a ``frontier of tractability.'' We also examine the relationship between social welfare and platform's revenue, and show that  where the platform can substantially increase social welfare, it
 can also extract substantial revenue. This yields an $O(\log(\min\{n, m\}))$-approximation algorithm for revenue in these kinds of inefficient
markets, as well as an upper bound of $O(\log(\min\{n, m\}))$ on the impact of platform self interest on welfare, relative to the optimal welfare in general markets. These bounds can be substantially improved in homogeneous-goods markets, where we prove that revenue maximization is perfectly aligned with welfare maximization, and give a constant-factor approximation algorithm for optimal revenue in  inefficient  markets. There are many interesting directions for future work:

\begin{itemize}
    \item The current gap between the lower and upper bounds for the effect of platform self interest
    on welfare (the {\em price of revenue maximization} or PRM) is quite large: $2$ vs $H_{\min(n, m)}=\Theta(\log_{\min(n, m)})$. We conjecture that $\mathrm{PRM} = 2$ in general markets. Closing this gap has important consequences for the impact of platforms on social welfare.
    \item  Can we resolve the complexity of  intermediate classes of markets; e.g.,  is it possible to exactly maximize revenue in homogeneous-goods markets in polynomial time? 
    \item We have assumed that sellers are indifferent between transacting off- and on-platform, whereas in real life sellers may choose to leave and transact with buyers off the platform if transaction fees are too high. It is interesting to model these kinds of considerations and understand which of our results still hold.
    \item Our model assumes full information with both valuations and world edges known to the platform. In reality, platforms only have partial information. Extending the model to a Bayesian setting, in which the platform faces uncertainty about buyer values and world edges, is also an interesting direction for follow up. 
\end{itemize}

\pagebreak
\bibliographystyle{ACM-Reference-Format}
\bibliography{main}

\newpage
\appendix
\section{Missing Proofs in Section~\ref{sec:hardness}}
\subsection{Proof of NP-Hardness (Lemma~\ref{lem:reduction_proof})}\label{app:sat_np_proof}
\SATReduction*
\begin{proof}
Suppose there is a valid assignment to the original CNF $\varphi$. Then there exists a matching between buyers $U_i$ and items $\alpha_{i,j}, \tau_{i,j}$ such that no buyer $U_i$ is matched to an item $\alpha_{i,j}$ where $x_i$ is false or an item $\tau_{i,j}$ when $x_i$ is true. Connect these as platform edges. For each buyer $A_{i,j}, T_{i,j}$, add platform edges to items $\gamma_{i,j}$ and $\delta_{i,j}$ respectively. Finally, add platform edges between each dummy item and a single unsold item so that all items are sold on platform.

The maximum weight matching assigns buyers $U_i$ their corresponding items via the platform edges. Each $A_{i,j}$ gets item $\gamma_{i,j}$ and each $T_{i,j}$ gets item $\delta_{i,j}$. Finally, the dummy buyers are matched to the remaining $(2tq - k)$ items that have not already been sold.
    
The total revenue of this matching is $D$. Each buyer $U_i$ pays price $Z$, resulting in revenue $kZ$. All dummy buyers pay price $H$, resulting in revenue $H(2tq-k)$. If $x_i$ is true, $\alpha_{i,0}$ is sold to a buyer $U_i$, so buyer $A_{i,j}$ pays price $Z$ as it has an opportunity path $A_{i,j} \mbox{ --- } \alpha_{i,0} \to U_i$ to this buyer. However, for $T_{i,j}$, both $\tau_{i,0}$ and $\tau_{i,j}$ are sold to dummy buyers – thus, $T_{i,j}$ only has opportunity paths to buyers with weakly higher values and must pay $Z + 1$. Thus, $\forall j, A_{i,j}$ has price $Z$ and $T_{i,j}$ has price $Z+1$. The converse holds when $x_i$ is False. In conclusion, from each $A_{i,j}$ and $B_{i,j}$, we get revenue $(t-1)Z + (t-1)(Z+1) = (t-1)(2Z+1)$. Thus, in total we get revenue $kZ + q(t-1)(2Z+1) + H(2tq - k) = D$ as desired. \\

Now, suppose that there exists a set of platform edges and a matching that generates revenue at least $D$. Then it must be the case that each buyer $U_i$ gets an item as otherwise we could have maximal revenue $(k-1)Z + 2q(t-1)(Z+1) + H(2tq - k) < D$. For each $U_i$ corresponding to clause $d_i = (x_i \vee \bar{x}_i)$, they must receive either $\alpha_{i,0}$ or $\tau_{i,0}$. We construct our satisfying assignment as follows:
\begin{itemize}
    \item If $U_i$ receives $\alpha_{i,0}$, set $x_i$ to True.
    \item If $U_i$ receives $\tau_{i,0}$, set $x_i$ to False.
\end{itemize}
It suffices to show that this is a satisfying assignment. To do this, we show that no buyer $U_k$ receives an item $\alpha_{i,j}$ when $x_i$ is False or an item $\tau_{i,j}$ when $x_i$ is True.

If $x_i$ is True, each $A_{i,j}$ has an opportunity path to buyer $U_i$, so we can get at most revenue $Z$ from them. Similarly, if $x_i$ is False, each $T_{i,j}$ has an opportunity path to buyer $U_i$, so we can get at most revenue $Z$ from them. Thus, the platform makes revenue at most $Z$ from each buyer $A_{i,j}$ or $B_{i,j}$ that corresponds to the $q$ true literal. For the $q$ false literals, note that the platform has to get maximal revenue $Z+1$ from the corresponding buyer because $D-kZ-H(2tq-k)-q(t-1)Z=q(t-1)(Z+1)$. If $U_k$ receives an item corresponding to the false literal, then the buyer $A_{i,j}$ or $T_{i,j}$ that corresponds to the false literal will have an opportunity path to this $U_k$, meaning they will only pay price $Z$, contradicting the above. It follows that our assignment is indeed a satisfying assignment as desired, concluding the proof.
\end{proof}

\subsection{Proof of APX-Hardness (Theorem \ref{thm:vc_apx_proof})}\label{app:vc_apx_proof}
The APX-hardness proof is similar to that in \citet{guruswami2005profit} for revenue maximizing envy-free pricing, and reduces from a version of the vertex cover problem. Given a connected graph $G=(V,E)$ where vertices have maximum degree $K$, the problem asks to find the minimum set of vertices to cover all edges. This problem is known to be APX-hard for $K \geq 3$. Given $G$, we construct the following instance of the platform's problem:
\begin{itemize}
    \item For each vertex $v \in V$, we create a vertex buyer $B_v$ and item $S_v$. We also create items $\alpha_{v, i}$ for $i \in \{1, \dots, \mathrm{deg}(v)\}$.
    \item For each edge $(u, v) \in E$, we create an edge buyer $B_{(u, v)}$.
    \item Finally, we add $|E|$ dummy buyers $D_1, \dots, D_{|E|}$.
    \item Each buyer $B_v$ has value $2$ for items $S_v, \alpha_{v, \cdot}$. Each buyer $B_{(u, v)}$ has value $1$ for items $\alpha_{u, \cdot}$ and $\alpha_{v, \cdot}$. Each dummy buyer $D_i$ has value $H\geq 2$ for all $\alpha$ items.
    \item We add world edges from $B_v$ to all $\alpha_{v, \cdot}$.
\end{itemize}
We complete the reduction by showing the following.
\begin{lemma}
\label{lem:vertice_cover_reduction}
    The optimal revenue is given by $2|V| + (H+1)|E| - q$ where $q$ is the size of the smallest vertex cover.
\end{lemma}
\begin{proof}
Suppose there is a vertex cover $Q$ of size $q$. We add platform edges from $B_v$ to $S_v$ for each $v \in V$. Additionally, for each $B_{(u, v)}$, if $u \in Q$, we add a platform edge to some $\alpha_{u, i}$. Otherwise – note that this implies that $v \in Q$ – we add a platform edge to $\alpha_{v, i}$. We do so such that each $B_{(u, v)}$ is matched via a platform edge to a different item. Finally, we add a platform edge from each $D_i$ to an $\alpha$ item that is not incident on any already added platform edges. Thus, the platform edges act as a perfect matching between buyers and items.

The platform then picks all the platform edges as the max weight matching and get revenue $2|V|+(H+1)|E|-q$. All dummy buyers pay full price $H$ and all edge buyers $B_{(u,v)}$ pays full price $1$ because they are not connected to any other buyers through opportunity path. Buyers $B_v$ where $v\in Q$ has an opportunity path $B_v \mbox{ --- } \alpha_{v,i} \to B_{v,u}$ to some edge buyer $B_{(u,v)}$ and pays price 1. Buyers $B_v$ where $v\notin Q$ pays full price $2$ because they are connected through opportunity path to dummy buyers. So total revenue from vertex buyers is $q+2(|V|-q)=2|V|-q$. \\

Now, we show the optimal revenue is bounded above by $2|V| + (H+1)|E| - q$ when the minimum vertex cover has size $q$. By Lemma~\ref{lem:at_most_one_edge}, if an edge buyer $B_{(u,v)}$ is matched, it must be matched to some item $\alpha_{u,i}$ or $\alpha_{v, i}$ that it has positive value for. If $B_{(u,v)}$ is not matched, Lemma~\ref{lem:all_transact} says all items $\alpha_{u, i}$ and $\alpha_{v, i}$ are matched with vertex buyer $B_u, B_v$ or some dummy buyers. No matter which buyer $\alpha_{u, i}$ and $\alpha_{v, i}$ are matched with, the platform can always match $B_v$ to $S_v$, $B_u$ to $S_u$, and the dummy buyers to some other $\alpha$ buyers because the number of buyers and sellers are equal. In doing so, the platform gain revenue 1 from $B_{u,v}$, and loses at most revenue 1 from $B_u$ or $B_v$. So we can assume there is a revenue-optimal matching where all edge buyers $B_{(u,v)}$ are matched with the corresponding $\alpha$ seller. 

As the minimum vertex cover is of size $q$, edge buyers are matched to some items $\alpha_{u, i}$ with at least $q$ different vertex $u$. Thus, from each corresponding vertex buyer $B_u$, we can attain at most revenue $1$ because they are connected through opportunity path to $B_{(u,v)}$. It follows that the maximum revenue we can achieve is given by $2|V| + (H+1)|E| - q$, concluding the proof.
\end{proof}

\VertexCoverReduction*
\begin{proof}
    We have assumed the graph that we reduce from is connected. So $|E| \geq |V| - 1$. Further because of maximum degree $K$, $|E|\leq |V|K/2$. So $|E|=\Theta(|V|)$. Again because of maximum degree $K$, the minimum vertex cover has size $q$ at least $|E|/K = \Omega(|V|)$. 
    
    Let $\mathrm{Rev}^\star$ be optimal revenue and $q^\star$ be the minimum vertex cover. By Lemma~\ref{lem:vertice_cover_reduction}, $\mathrm{Rev}^\star=2|V|+(H+1)|E|-q^\star$. Then a constant factor approximation for $\mathrm{Rev}^\star$ translates into a constant factor approximation for $q^\star$, yielding a PTAS reduction and concluding the proof of APX-hardness.
\end{proof}

\section{Missing Proofs in Section~\ref{sec:special-case-homogeneous}}

\subsection{Proof of Theorem \ref{thm:opt_char}}
\label{app:swsh_char}

Here, we provide the proof of Theorem~\ref{thm:opt_char}, characterizing the optimal platform matching in SWSH markets.

\SWSHChar*

We start with the following lemma, which places a constraint of the length of any cycles in the optimal solution.
\begin{lemma}
    There exists an optimal set of transactions that can be decomposed into chains and/or cycles. Furthermore, there exists an optimal solution where all such cycles are of length at most $3$.
\end{lemma}
\begin{proof}
    Suppose that we had a cycle of length at least $4$. Then we could always split this cycle into smaller cycles of length at most $3$ (note that we can represent every integer larger than $3$ as a sum of multiples of $2$ and $3$). Furthermore, these cycles can only result in weakly higher revenue among sellers included in the original cycle as we are decreasing the number of opportunity paths. Finally, any chain connected to the original cycle can now be attached to the same subgraph in this new split set of cycles – again, there are fewer opportunity paths, so each seller in the chain also obtains weakly higher revenue. 

     It follows that there always exists an optimal solution where all cycles have length at most $3$.
\end{proof}

 We make the following simple observation about the optimal configuration of such cycles and chains.

\begin{remark}
    The cycle on $k$ subgraphs that yields optimal revenue connects the highest buyer in each subgraph together. The optimal-revenue chain connects subgraphs from smallest to largest in terms of their max-value bidder. The chain and cycle that yields optimal revenue is constructed by attaching the optimal chain to the largest second-highest bidder in any subgraph in the optimal cycle.
\end{remark}

 In fact, we can also show that there exists an optimal solution with only a single chain.
\begin{lemma}
    There exists an optimal solution with only a single chain. Additionally, this optimal solution can still restrict the cycle length to at most $3$.
\end{lemma}
\begin{proof}
    Suppose we had two chains, each attached to a different cycle, or potentially attached to the same cycle. We could always consolidate these two chains into a single chain and attach it to the cycle that yields the largest minimum opportunity path – all sellers in the chain now receive weakly more revenue than they were before. This is doable because each seller subgraph has at least one seller and one buyer to be connected into a chain. Note that this doesn't affect the restriction on cycle length. 
\end{proof}
To conclude the discussion of the structure of chains in the optimal solution, we show that these chains must be contiguous and moreover, chains ``fill out" the rest of the subgraphs once started.
\begin{lemma}
    If $S_i$ is part of a chain in the optimal solution, then $S_j$ is also part of this chain, for $j > i$, provided that $S_j$ is not part of the cycle that the chain attaches to.
\end{lemma}
\begin{proof}
    Suppose otherwise. Consider the optimal solution with at most one chain. Since $S_j$ is not part of a chain, it must be part of a cycle. As this solution is optimal, it must be the case that adding this whole cycle to the chain decreases the overall revenue. That is, the minimum opportunity path in $S_j$'s cycle is larger than the minimum opportunity path in the cycle that $S_i$'s chain is attached to. 

     However, this implies that we could add $S_i$ to $S_j$'s cycle, increasing the revenue generated from the seller in $S_i$ without affecting the revenue from any of the sellers in $S_j$ (as $S_i$ has a larger max-value buyer than $S_j$), which is a contradiction.
\end{proof}
\begin{lemma}
    Suppose that there is a chain in the optimal solution that connects to a cycle $\mathcal{C}$, and let $S_{min, \mathcal{C}}$ be the smallest subgraph belonging to $\mathcal{C}$. Then there exists an optimal solution where no subgraph larger than $S_{min, \mathcal{C}}$ is part of the chain.
\end{lemma}
\begin{proof}
    Suppose otherwise. If there are two or more such subgraphs, then we could always connect them in a cycle, and their minimum opportunity path would be weakly larger than that obtained by placing them in the chain. If there is a single such subgraph, we could add it to $\mathcal{C}$ – note that this does not affect the minimum opportunity path of the chain or of the sellers in $\mathcal{C}$ and weakly increases the revenue obtained from this subgraph. If at any point during this process, we obtain a cycle with more than three subgraphs, we can simply split it up into smaller cycles, preserving our desired property.
\end{proof}

Note that this implies that once a subgraph is part of a chain, all smaller subgraphs (in terms of max bidder value) must also be part of this chain. We make one final observation regarding the structure of the optimal solution – namely, all cycles in the optimal solution can be made to be contiguous.
\begin{lemma}
    There exists an optimal solution where all cycles are contiguous.
\end{lemma}
\begin{proof}
    Suppose we have an optimal solution satisfying all the above properties where at least one cycle is non-contiguous. Let $S_m$ be the smallest subgraph that belongs to a non-contiguous cycle (clearly $S_m \neq S_1$).

     Firstly, if $S_m$ is a part of a three-cycle $S_a - S_b -S_m$. If $a<b-1$, then we can connect $S_b - S_m$ in a two cycle. By Lemma 4.5, $S_{b-1}$ cannot be in the chain otherwise $S_b$ and $S_m$ would be in the chain. Thus, $S_{b-1}$ is in a cycle. We attach $S_a$ to the cycle that $S_{b-1}$ is in. This weakly increases $S_a$'s revenue because $S_{b-1}$ only connects to larger seller subgraphs.    

     Now we can focus on non-contiguous two-cycles $S_a - S_m$ and three-cycles $S_a - S_{a+1} - S_m$.
    Again, by Lemma 4.5, $S_{m-1}$ cannot be a part of the chain otherwise $S_m$ would have been in the chain as well.
    There are two cases. Suppose that the second highest bidder in $S_{m-1}$ has a weakly higher valuation than the highest bidder in $S_m$. Consider the following alternative configuration. Connect $S_m$ to $S_{m-1}$ as a chain. 
    (If a chain is already connected to $S_{m-1}$, add $S_m$ to this chain.)
    Combine the two cycles that $S_{m-1}, S_m$ belonged to (save for $S_m$, works also if $S_{m-1}$ is in a 1-cycle.) Additionally, keep any existing chains that were connected to any of the subgraphs involved in the two cycles. 
    
     Note that all subgraphs that were previously connected to $S_m$ via a cycle now have weakly increased revenue – the minimum opportunity path in the new cycle is now $S_{m-1}$ rather than $S_m$. All subgraphs that were previously connected to $S_{m-1}$ have the same revenue; all such subgraphs still have $S_{m-1}$ as their smallest opportunity path within the cycle. It is easy to verify that all chains also have weakly higher revenue. 

     Now, consider the second case where the second highest buyer in $S_{m-1}$ has a smaller valuation than the highest bidder in $S_m$, or $S_{m-1}$ does not have a second world buyer. There are now a couple of possible cases:
    \begin{itemize}
        \item Two or more subgraphs were previously connected to $S_m, S_{m-1}$. Connect $S_m, S_{m-1}$ in a two-cycle. In this case, connect all the subgraphs that were previously connected to $S_m, S_{m-1}$ together into a single cycle. Note that we lose $v(S_{m-1}) - v(S_m)$ in revenue from connecting $S_m$ and $S_{m-1}$ together, but we gain at least $v(S_{m-2}) - v(S_m)$ from the other subgraphs.
        \item $S_{m}$ belonged to a 2-cycle and $S_{m-1}$ belonged to a 1-cycle. If the other subgraph in the $S_m$ cycle was $S_{m-2}$, then connect $S_{m-2}, S_{m-1}$ together and leave $S_m$ as a 1-cycle. Note that we get weakly more revenue from $S_{m-2}$ in this case, and a 1-cycle with $S_m$ is less costly than a 1-cycle with $S_{m-1}$.
        \item $S_{m}$ belonged to a 2-cycle and $S_{m-1}$ belonged to a 1-cycle. Connect $S_m, S_{m-1}$ in a two-cycle. If the other subgraph in the $S_m$ cycle was not $S_{m-2}$, connect the other subgraph in $S_m$'s cycle to the cycle that $S_{m-2}$ is currently a part of. Note that we gain $v(S_m)$ from connecting $S_m$ and $S_{m-1}$ together, and at least $v(S_{m-2}) - v(S_m)$ from the other subgraphs.
    \end{itemize}
     In all the above cases, we keep all existing chains as they are – they are weakly better off under this new configuration because the second highest buyer in $S_{m-1}$ is smaller than $v(S_m)$. We continue this process until there are no non-contiguous cycles, at each step weakly increasing our revenue. If at step $t$, the smallest subgraph belonging to a non-contiguous cycle is $S_m$, then at step $t+1$, the process ensures that the smallest possible subgraph belonging to a non-contiguous cycle is $S_{m-1}$, meaning that this process must terminate.
\end{proof}

Combining the above lemmas, we have proved Theorem \ref{thm:opt_char}, which we restate below for completeness.

\SWSHChar*

\subsection{Proof of Theorem \ref{thm:swsh_reduction}}
\label{app:swsh_reduction}

Here, we complete the proof of Theorem~\ref{thm:swsh_reduction}, which states that there exists a polynomial-time algorithm to maximize the platform's revenue in general SWSH markets.

\SWSHReduction*

We prove the above theorem through two lemmas, addressing the case where $n > m$ and $n < m$ respectively. Note that the case in which $n = m$ is addressed by Theorem~\ref{thm: poly_AMOS}.

\begin{lemma}\label{lem:hom_sellers_transact_with_largest_buyers}
    When $n=|B|>|S|=m$, sort buyers by their valuation, and denote the set of the $m$ largest buyers by $B_m=\{b_1,b_2,...,b_m\}$. Then perform the revenue-optimal matching on $S, B_m$ while ignoring $B\setminus B_m$. This is revenue optimal.
\end{lemma}

\begin{proof}
    By Lemma \ref{lem:all_transact}, there exists a revenue-optimal transaction where all sellers trade. Otherwise, all sellers connecting to this non-trading seller through an opportunity path have a price of zero. By adding an edge between this seller and a non-trading buyer (guaranteed to exist as $n > m$), other sellers' prices weakly increase.
    
     Further, no buyer in $B\setminus B_m$ trades. Suppose otherwise and take any configuration where a buyer in $B \setminus B_m$ trades. Let $b_{min}$ be the smallest trading buyer, and let $b_{max}$ be the largest non-trading buyer, noting that $v(b_{min}) < v(b_{max})$. Suppose that $b_{min}$ currently transacts with seller $s$. 

     As the market clears according to the maximum weight matching, it must be that $s$ is not connected to $b_{max}$ – otherwise, matching them would yield a strictly higher welfare matching. Thus, we can add a platform edge from $s$ to $b_{max}$, ensuring that they transact in the new matching. Note that $b_{min}$ no longer transacts as it was the minimum value transacting buyer, and thus the maximum welfare matching can no longer include it. 
    
     We claim that all sellers' prices weakly increase. The only modified opportunity paths that previously existed are those which previously went through $b_{min}$ – as previously mentioned, there are no opportunity paths that terminate at a non-trading seller, so this was the minimum possible opportunity path. Now, no opportunity paths terminate at $b_{min}$ (as they no longer trade), meaning that the minimum opportunity path weakly increased (in the event of a tie with $b_{min}$). 

     We can continue this process until no buyers in $B \setminus B_m$ trade. Buyers not trading will not be in any opportunity path hence not affect the characterization results nor optimality of the procedure in Theorem~\ref{thm: poly_AMOS}.
\end{proof}

\begin{lemma}
\label{lemma:discard_dangling_sellers}
    When $n = |B| < |S| = m$, discard $m - n$ of the sellers who do not have world edges, and perform the revenue-optimal matching on the induced graph. This is revenue optimal.
\end{lemma}
\begin{proof}
    By Lemma \ref{lem:all_transact}, note that there exists a revenue optimal matching in which all buyers transact. Suppose that there is a seller $s$ who has at least one adjacent world edge who does not transact. Consider any buyer $b$ who is in seller $s$'s subgraph. 

     Clearly, buyer $b$ must transact with some seller $s'$ via a platform edge (as they do not transact with the seller in their subgraph). Consider eliminating this platform edge $(b, s')$. No revenue is lost from $b$ as they had an opportunity path to $s$, who did not transact, so they paid price $0$. Secondly, any sellers who previously had opportunity paths to $b$ also received price $0$, so removing this edge does not decrease their price either. 

     It follows that there exists a revenue optimal matching where all sellers with an adjacent platform edge transact, which implies that we can eliminate $m-n$ sellers with no adjacent world edges, as they would not transact in this revenue optimal matching regardless.
\end{proof}

Together, Lemmas~\ref{lem:hom_sellers_transact_with_largest_buyers} and \ref{lemma:discard_dangling_sellers} show that general SWSH markets can be reduced to the case in which $n = m$. Combining this with Theorem~\ref{thm: poly_AMOS} concludes the proof of Theorem~\ref{thm:swsh_reduction}.

\section{Missing Proofs in Section~\ref{sec:special-case-identical}}
\subsection{Proof of Lemma \ref{lemma:identity_char}}
\label{app:hv_identity}

Here, we prove Lemma~\ref{lemma:identity_char}, connecting the maximum revenue attainable in an identity-good market to the problem of finding buyer sets of non-positive surplus.

\IdentityChar*

We begin by first proving the following lemma, which characterizes the revenue attainable between a particular set of buyers and a set of sellers.

\begin{lemma}\label{cor:price_k_hall_violator}
    Given $k$ buyers and sellers $B_k=\{b_1,b_2,...,b_k\}, S_k=\{s_1,s_2,...,s_k\}$ where $B_k$ and $S_k$ have a perfect matching through only platform edges, total revenue $k$ can be attained if and only if $B_k$ belongs to a Hall violator $B_{vk}$ for graph $G^{-k}=(S\setminus S_k \cup B, E)$ of deficiency at least $k$.
\end{lemma}
\begin{proof}
Assume that $B_k$ belongs to a Hall violator $B_{vk}$ for graph $G^{-k}=(S\setminus S_k \cup B, E)$ of deficiency at least $k$. Since $B_{vk}$ has deficiency at least $k$, there exists a set of platform edges such that in the final maximum weight matching, $S_k$ transacts with $B_k$ and all sellers in $N(B_{vk})$ transact with buyers in $B_{vk}$. In this case, all of the buyers in $B_k$ only have opportunity paths to those sellers in $N(B_{vk})$, who all transact, so the platform obtains revenue $1$ from all these buyers, yielding revenue $k$ in total.

Now, suppose that total revenue $k$ can be attained from $B_k$ and $S_k$. Let $G_P^{-k}$ denote the corresponding platform graph, with the set of sellers $S_k$ removed. Let $B_{vk}$ be the set of all buyers on an opportunity path from any of the buyers in $B_k$. Note that all the sellers in $N_{G_P^{-k}}(B_{vk})$ must transact as otherwise at least one buyer in $B_k$ would transact at price zero, thus yielding less than $k$ revenue in total. Additionally, all these sellers must transact with a buyer in $B_{vk}$ as they are reachable via an opportunity path from $B_k$. Furthermore, all the sellers in $B_k$ transact with $S_k$. Thus, we have that $|B_{vk}| \geq |N_{G_P^{-k}}(B_{vk})| + k$. By Lemma \ref{lem:all_transact}, all platform edges transact, so it is also the case that $|B_{vk}| \geq |N_{G^{-k}}(B_{vk})| + k$, concluding the proof.
\end{proof}

Given $B_k, S_k$, finding such a Hall violator is equivalent to the max difference hall violator problem defined in Definition~\ref{def:max_dif_hall_violator}. The question for a revenue-optimizing platform is then to pick the sets $B_k, S_k$. The platform can first identify the buyers as follow.

\begin{corollary}\label{cor:price_k_hall_violator_buyer}
    The platform can extract revenue $k$ from a set of buyers $B_k$ if and only if there exists a buyer group $B_v$ where $B_k\subseteq B_v$ that satisfies the following two conditions. Let $E^{\neg}$ be the complement edges of the world graph, let $k_v\leq k$ be the value of maximum matching for $(B_k, N(B_v), E^{\neg})$.
    \begin{enumerate}
        \item Enough sellers to match with $B_k$ buyers $|S|-|N(B_v)|\geq k-k_v$
        \item $B_v$ has deficiency at least $k-k_v$
    \end{enumerate}
\end{corollary}
\begin{proof}
    If direction. Match the $k_v$ sellers in $N(B_v)$ through platform edges to $k_v$ buyers in $B_v$. For the rest of the $k-k_v$ buyers, match them with arbitrary sellers outside of $N(B_v)$. Denote the union of these two sets of sellers by $S_k$. Then $B_v$ is a Hall violator for graph $G^{-k}=(S\setminus S_k\cap B,E)$ of deficiency at least $k$. This satisfies Lemma~\ref{cor:price_k_hall_violator}. 
    
    Only if direction. Given a seller set $S_k$ that sells to $B_k$, let $B_v$ be the set of buyers reachable through opportunity paths from $B_k$, unioned with $B_k$ itself. $B_v$ satisfies the first criterion: at most $k_v$ buyers from $B_k$ can be matched with sellers from $N(B_v)$, the rest of the sellers are from $S\setminus N(B_v)$. All sellers in $N(B_v)$ transact with a unique buyer in $B_v$. If only $k'$ such sellers transact with $k'$ buyers in $B_k$, the rest of the $k-k'$ buyers must transact with sellers in $S\setminus N(B_v)$. Thus, $|B_{v}|=|N(B_{v})|+(k-k')$ and $B_v$ has deficiency $k-k'$. As $k'\leq k_v$, $B_v$ has deficiency weakly larger than $k-k_v$ in the platform graph. This does not change in the original world graph.
\end{proof}

\begin{corollary}\label{cor:price_k_hall_violator_vertex_set_revenue}
    Given a buyer set $B_v$ with non-positive surplus in $G$, 
    let $E^{\neg}$ be the complement edges of the world graph, and let $k_v$ be the value of maximum matching for $(B_v, N(B_v), E^{\neg})$. The platform can extract $$k = k_v + \min\{|B_v|, |S|\}-|N(B_v)|$$
    revenue from the buyers in $B_v$, and this is the maximum possible revenue achievable from the buyers in $B_v$. 
\end{corollary} 
\begin{proof}
    Firstly, note that the stated revenue is weakly smaller than the number of buyers in $B_v$: $k\leq k_v+|B_v|-|N(B_v)|\leq |B_v|$ since $k_v\leq |N(B_v)|$. We pick the buyer set $B_k$ as follows: 1) Take $k_v$ buyers in the maximum matching in $(B_v, N(B_v), E^{\neg})$, 2) Additionally, take $\min\{|B_v|, |S|\}-|N(B_v)|\geq 0$ buyers in $B_v$ besides the $k_v$ ones. We show that $B_k$ satisfies the two conditions in Corollary ~\ref{cor:price_k_hall_violator_buyer}.
    \begin{eqnarray*}
        k-k_v = \min\{|B_v|, |S|\}-|N(B_v)| =  \min\{|B_v|-N(B_v), |S|-N(B_v)\}
    \end{eqnarray*}
    The first and second terms in the min correspond to the first and second conditions of Corollary~\ref{cor:price_k_hall_violator_buyer} respectively. Additionally, a larger $k$ will violate at least one of the conditions so $k$ is the maximum attainable revenue.
\end{proof}

Combining the above results, we state and prove Lemma~\ref{lemma:identity_char}.

\IdentityChar*
\begin{proof}
    Suppose that the optimal revenue is $k$ and the platform matches $S_k$ to $B_k$. Corollary~\ref{cor:price_k_hall_violator_buyer} states that there is a set $B_k\subseteq B_v$ with non-positive surplus. Corollary~\ref{cor:price_k_hall_violator_vertex_set_revenue} gives the max revenue from $B_v$. As $k$ is the optimal revenue in the entire graph, it is also the maximum revenue from $B_v$. Thus, $k=k_v + \min\{|B_v|, |S|\}-|N(B_v)|$. If $k\geq x$, then $k_v + \min\{|B_v|, |S|\}-|N(B_v)|\geq x$. Conversely, if for all buyer sets $B_v$ with non-positive surplus, $k_v + \min\{|B_v|, |S|\}-|N(B_v)|< x$, then the optimal revenue is given by $k<x$.
\end{proof}

\subsection{Proof of Lemma \ref{lemma:max_cardinality}}
\label{app:max_cardinality}

In this section, we give the proof of Lemma~\ref{lemma:max_cardinality}, simplifying the platform's problem in SHGB markets to the problem of finding the maximum cardinality set of buyers with non-positive surplus. To do so, we first show in sparse graph with buyer degree at most 2, any set of buyers $B_v$ with non-positive surplus satisfies $k_v = |N(B_v)|$, except for certain special cases.

\begin{lemma}
\label{lemma:full_matching}
    Consider any set of buyers $B_v$ with non-positive surplus in SHGB markets. if $|B_v| \geq 3$ and there are more than 2 degree-two buyers in $B_v$, then $k_v = |N(B_v)|$.
\end{lemma}
\begin{proof}
    Construct the flow network $G^{\neg}_{\infty}$  from the complement graph  $G^{\neg}=(B_v, N(B_v),E^{\neg})$ in such a way that there are edges of capacity 1 from the source $p$ to every vertex in $N(B_v)$ and from every vertex in $B_v$ to the sink $q$, and of capacity $+\infty$ from $s$ to $b$ for any $(s,b)\in E^{\neg}$.
    \begin{figure}[!h]
        \centering
        \includegraphics[width=0.5\textwidth]{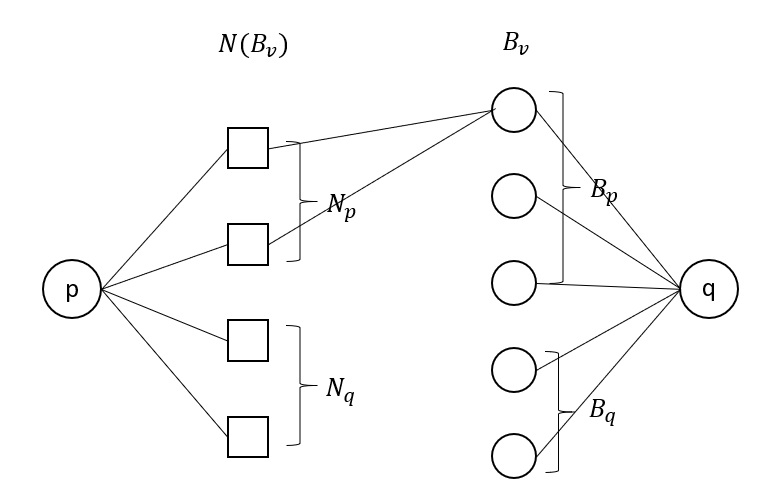}
    \end{figure}
    The size of the maximum matching in $G^{\neg}$ equals to that of the max flow/min cut in $G^{\neg}_{\infty}$. Let $(P,Q)$ be a min cut and $C\leq |N(B_v)|$ be the cut value. Let $N(B_v)=N_{p}\cup N_{q}, B_v=B_{p}\cup B_{q}$, such that $N_{p},B_{p}\subset P, N_{q}, B_{q}\subset Q$. The minimum cut is composed of edges going from $p$ to $N_{q}$ and from $B_{p}$ to $q$, as any edges between $N_q, B_p$ and $N_p, B_q$ would make the size of the cut infinite.

    We proceed by casework:
    \begin{itemize}
        \item Firstly, note that if any degree-zero buyer is in $B_q$, then it must be that $N_p = \emptyset$ as degree-zero buyers are adjacent to all sellers in $N(B_v)$ in the complement graph. Thus, $C \geq |N_q| = |N(B_v)|$.
        \item If $|B_q| = 0$, then we have that $C \geq |B_v| \geq |N(B_v)|$.
        \item Suppose that $|B_q| = 1$. As all buyers have degree at most $2$, it must be that $|N_p| \leq 2$. thus, we have that $C \geq |N_q| + |B_p| \geq (|N(B_v)| - 2) + (|B_v| - 1) \geq |N(B_v)|$ since $|B_v| \geq 3$.
        \item Suppose that $|B_q| = 2$. In any case, it must be that $|N_p| \leq 1$ since either we have a degree-one buyer present in $B_q$ or two degree-two buyers (in which case they cannot share the same two sellers). Thus, $C = |N_q| + |B_p| \geq  (|N(B_v)| - 1) + 1 = |N(B_v)|$.
        \item Suppose that $|B_q| \geq 3$. If there are at least $3$ degree-two buyers in $B_q$, then since a pair of sellers knows at most one common buyer, it follows that $|N_p| = 0$. Thus, we have that $C \geq |N(B_v)|$.

        Otherwise, there must be at least one degree-one buyer in $B_q$ so $|N_p| \leq 1$. Additionally, at least one degree-two buyer must be excluded from $B_q$. Thus, we have that $C = |N_q| + |B_p| \geq  (|N(B_v)| - 1) + 1 = |N(B_v)|$.
    \end{itemize}

    Thus, if $|B_v| \geq 3$ and there are more than $2$ degree-two buyers in $B_v$, it follows that $k_v = |N(B_v)|$, concluding the proof.
\end{proof}

We need one more lemma before proving Lemma~\ref{lemma:max_cardinality}. Lemma~\ref{lem:all_degree} shows that there exists an optimal set of buyers $B_v$, in the sense of Lemma~\ref{lemma:identity_char}, in which all buyers of degree zero and degree one are included.

\begin{lemma}
\label{lem:all_degree}
    There exists a set of buyers $B_v$ maximizing
    \begin{align*}
        \min(|B_v|, |S|) - |N(B_v)| + k_v
    \end{align*}
    such that all buyers of degree zero and degree one are included in $B_v$.
\end{lemma}
\begin{proof}
    Clearly, any degree-zero buyers can only contribute positively to the above expression. Thus, all buyers of degree zero are WLOG in the optimal set $B_v$. 
    
    Turning to the case of buyers of degree one, consider any alternative set of buyers $B'_v$. There are two cases for the change in the above expression when adding a buyer of degree one. Let $s_j$ be the seller adjacent to this buyer. If $s_j \in N(B'_v)$, note that the above expression can only increase.

    If $s_j \not \in N(B'_v)$, then there are again two cases. If $|B'_v| < |S|$, then the above expression can again only increase – there is an increase of $1$ from the first term and a decrease of $1$ from the second term. If $|B'_v| \geq |S|$, then since $s_j \not \in N(B'_v)$, we claim that the $k_v$ term must strictly increase. Consider any previous maximum matching between $B'_v$ and $N(B'_v)$. Take an arbitrary buyer $b'$ in $B'_v$ (note that one must exist since $|B'_v| \geq |S|$) and match them to $s_j$ while matching the new buyer to the seller $b'$ was matched to in the maximum matching (if one exists). This increases the cardinality of the maximum matching by $1$, so it follows that the value of above expression does not decrease. This concludes the proof.
\end{proof}

Given the above lemma, we now prove Lemma~\ref{lemma:max_cardinality}.

\MaxCard*
\begin{proof}
    Restricting our search to sets of buyers $B_v$ that satisfy the conditions of Lemma~\ref{lemma:full_matching}, we have that $k_v = |N(B_v)|$. This corresponds to a set of $B_v$ such that $|B_v| \geq 3$ and there are more than 2 degree-two buyers in $B_v$. By Lemma~\ref{lemma:identity_char}, the revenue from this given set of buyers $B_v$ is given by
    \begin{align*}
        \min(|B_v|, |S|) - |N(B_v)| + k_v = \min(|B_v|, |S|) - |N(B_v)| + |N(B_v)| = \min(|B_v|, |S|)
    \end{align*}
    so maximizing the platform's revenue over these sets is equivalent to finding the maximum cardinality such set of buyers.

    Now, consider the sets of buyers $B_v$ that do not satisfy the conditions of Lemma~\ref{lemma:full_matching}. These are sets of buyers with non-positive surplus $B_v$ such that either $|B_v| < 3$ or there are no more than $2$ degree-two buyers in $B_v$. Clearly, we can enumerate all $B_v$ such that $|B_v| < 3$ and check the revenue attainable from each of these sets. For those $B_v$ containing no more than $2$ degree-two buyers in $B_v$, Lemma~\ref{lem:all_degree} proves that it is WLOG to begin with all degree zero and degree one buyers in $B_v$. Thus, we can simply ennumerate over all $O\left(\binom{n}{2}\right)$ choices of two degree-two buyers to add, check the revenue attainable from each of these sets. This brute force check takes polynomial time, concluding the proof.
\end{proof}

\subsection{Proof of Theorem~\ref{thm:poly_identity}}
\label{app:poly_identity}

Here, we provide the proof of Theorem~\ref{thm:poly_identity}, showing that we can maximize the platform's revenue in SHGB markets in polynomial time. We begin with the following lemma, which will form the crux of our proof.

\begin{lemma}
\label{lemma:induced_edges}
    Consider a world graph where all buyers are of degree two and for every pair of sellers, at most one buyer knows them both. The maximum cardinality set $B_v$ of buyers with surplus at most $k$ – that is, $|N(B_v)| - |B_v| \leq k$ – can be identified in polynomial time.
\end{lemma}
\begin{proof}
    Note that such world graphs are in direct correspondence with general graphs $G = (V, E)$ where each buyer corresponds to an edge and each seller corresponds to a vertex. Specifically, $b_i$ corresponds to an edge $(s_i, s_j)$ between the vertices representing sellers $s_i, s_j$ whom $b_i$ is connected to in the original world graph.

    Restated in these terms, finding the maximum cardinality set $B_v$ of buyers with surplus at most $k$ is equivalent to finding the maximum cardinality set of edges $\tilde{E} \subseteq E$ such that the induced subgraph of $G$ corresponding to $\tilde{E}$ has at most $|\tilde{E}| + k$ vertices.

    We claim that $\tilde{E}$ must correspond to a union of connected components of $G$. Suppose that $\tilde{E}$ contained only a subset of edges within a given connected component $C$ of $G$. By adding all the edges in $C$, the number of vertices in the induced subgraph can increase by at most the number of additional edges, not harming the surplus of the induced subgraph in the process. It follows that $\tilde{E}$ must correspond to a union of connected components of $G$ as desired.

    Thus, our algorithm to find the maximum cardinality set of edges $\tilde{E} \subseteq E$ such that the induced subgraph of $G$ corresponding to $\tilde{E}$ has at most $|\tilde{E}| + k$ vertices proceeds as follows:
    \begin{itemize}
        \item Find all the connected components of $G$. If a connected component is not a tree, it contributes more edges than vertices, so we can add it to $E_{cur}$.
        \item Having added all non-tree connected components, compute the difference between the number of edges and number of vertices in the induced subgraph of $E_{cur}$. Call this difference $\ell$.
        \item Take the $k + \ell$ largest (in terms of the number of edges) connected components of $G$ that are trees and add their edges to $E_{cur}$. Return the final set as $\tilde{E}$.
    \end{itemize}
    We now prove the correctness of our algorithm below. Clearly, $\tilde{E}$ must contain all connected components of $G$ that are not trees. Given this fact, the remaining edges in $\tilde{E}$ must correspond to connected components of $G$ that are trees. Our algorithm takes precisely those trees with the largest number of edges, completing the proof.
\end{proof}

\PolyIdentity*
\begin{proof}
    It suffices to show that we can find the maximum cardinality set $B_v$ of buyers of non-positive surplus in polynomial time in SHGB markets, which together with Lemma~\ref{lemma:max_cardinality} implies the desired result.

    All buyers of degree zero and degree one contribute at most one to $|N(B_v)|$ while contributing at least one to $|B_v|$. Thus, they must all be included. Upon including all these buyers, consider the induced subgraph obtained by removing all sellers adjacent to a buyer of degree one or zero. This induced subgraph may have more buyers of degree one or zero, which by the same argument as above, must be included in $B_v$. We can repeat this process until we are left with a current set of buyers $B'_v$ which must be included in $B_v$ as well as a graph in which all buyers are of degree two and for every pair of sellers, at most one buyer knows them both.

    Thus, our problem reduces to finding the maximum cardinality set of these buyers with surplus at most $|B'_v| - |N(B'_v)|$. Lemma~\ref{lemma:induced_edges} shows that this is solvable in polynomial time, concluding the proof.
\end{proof}

\section{Missing Proofs in Section~\ref{sec:hom_markets}}
\subsection{Proof of Theorem~\ref{thm:hom_min_nm_approx}}\label{app:min_nm_approx_rev}
Here, we provide the proof of the following theorem, giving a $\min\{n, m\}$ approximation algorithm for revenue in general homogeneous-goods markets.
\HomoMinNM*

In order to guarantee this approximation ratio, we consider a buyer-seller pair $(b_i, s_j)$ and aim to find the maximum possible revenue the platform could extract from the single transaction between these two agents. Asking this same question for all possible pairs, we are guaranteed to obtain the maximum possible revenue the platform could attain from any single transaction. As there at most $\min\{n, m\}$ transactions conducted via the platform, being able to achieve this maximum transaction in polynomial time would prove Theorem~\ref{thm:hom_min_nm_approx}.

We begin by noting that we can restrict our attention to transactions involving buyers among those with the top $\min\{n, m\}$ highest valuations.

\begin{lemma}
    Up to ties, the buyer-seller pair $(b_i, s_j)$ from which the platform can extract the maximum revenue always involves a buyer among those with the top $\min\{n, m\}$ highest valuations. Additionally, the maximum revenue is either zero or at least $v_{\min\{n, m\}}$, where this is the $\min\{n, m\}^{th}$ largest valuation.
\end{lemma}
\begin{proof}
    The statement is trivial if $m \geq n$. Otherwise, if $n > m$, then either the top $m$ buyers already have all possible world edges, in which case we can extract $0$ revenue from the market or there exists a buyer among the top $m$ buyers that is missing at least one world edge. In this case, we can add this as a platform edge and match off the remaining top $m - 1$ buyers, guaranteeing us a revenue of at least $v_{(m)}$, where this is the $m^{th}$ largest valuation. As we can extract at most $v_i$ from buyer $b_i$, it follows that the revenue-optimal pair $(b_i, s_j)$ must include one of the top $m$ buyers.
\end{proof}

We will show that the maximum price for a $(b_i, s_j)$ pair is closely connected to the notion of Hall violators, most notable for their use in matching theory \citep{lovasz2009matching}.

\begin{definition}[Hall violator]
    A set $B_v$ of buyers is a Hall violator if $|B_v| > |N(B_v)|$, where $N(B_v)$ is the set of sellers connected to any buyer in $B_v$ via a world edge.
\end{definition}

Next, we state the relationship between Hall violators and the maximum price attainable from a buyer-seller pair. Let $\tilde{B}$ denote the set of buyers who have one of the top $\min\{n, m\}$ valuations. Let $v_{(\min\{n, m\})}$ denote the $\min\{n, m\}^{th}$ largest valuation. By Lemma 

\begin{lemma}
\label{lemma:hall_violator_prices}
    Consider a buyer-seller pair $(b_i, s_j)$, where $b_i \in \tilde{B}$ is not connected to $s_j$ via a world edge. There exists a set of platform edges the platform can add such that $s_j$ sells to $b_i$ at a positive price $k \geq v_{(\min\{n, m\})}$ if and only if $b_i$ belongs to a Hall violator $B_{vi}\subseteq \tilde{B}$ for $G^{-j} = (S \setminus \{s_j\} \cup \tilde{B}, E)$ where $k \leq \min_{b_i \in B_{vi}} v_i$.
\end{lemma}
\begin{proof}
    Suppose there exists a set of platform edges such that $b_i$ transacts with $s_j$ at a positive price $k \geq v_{(\min\{n, m\})}$. Consider the set of all buyers reachable via an alternating/opportunity path from $b_i$. As $k > 0$, this set of buyers $B_{vi}$ forms a Hall violator in $G^{-j}$; otherwise, there would be an opportunity path from $b_i$ to an unsold item. Additionally, $s_j$ transacts at price $\min_{b \in B_{vi}} v(b)$ by Theorem~\ref{thm:oppo_path}. Since $k \geq v_{(\min\{n, m\})}$, $B_{vi}$ cannot contain any buyers outside of $\tilde{B}$. 

    Now, suppose we have some Hall violator $B_{vi} \subseteq \tilde{B}$ of $G^{-j}$ that includes $b_i$, where $k \leq \min_{b_i \in B_{vi}} v_i$. We claim that the platform can add edges $E_P$ and induce platform graph $G_P=(S\cup B,E\cup E_P)$ such that in the max weight matching:
    \begin{itemize}
        \item Sellers in $N_{G_P}(B_{vi})$ all transact.
        \item The sellers in $N_{G_P}(B_{vi})$ only transact with the buyers in $B_{vi} \setminus b_i$.
        \item $s_j$ transacts with $b_i$.
    \end{itemize}

    Consider the graph $G^{-j}$. As we have a Hall violator $B_{vi}\subseteq \tilde{B}$ of this graph, it follows that $|N_{G^{-j}}(B_{vi})| < |B_{vi}| \implies |N_{G^{-j}}(B_{vi})| \leq |B_{vi} \setminus \{b_i\}|$. Note that any matching that perfectly matches the top $\min\{m, n\}$ buyers to the top $\min\{m, n\}$ sellers is a maximum weight matching.

    Thus, we add the following set of platform edges: connect $b_i$ to $s_j$, connect the largest $|N_{G^{-j}}(B_{vi})|$ buyers in $B_{vi} \setminus \{b_i\}$ to $N_{G^{-j}}(B_{vi})$, and connect any extra buyers in $\tilde{B}$ arbitrarily to the remaining sellers not in $N_{G^{-j}}(B_{vi})$. We can then arbitrarily break ties in favor of the max weight matching that matches $b_i$ with $s_j$, that matches $N_{G^{-j}}(B_{vi})$ with the buyers in $B_{vi} \setminus \{b_i\}$, and that matches the rest of the buyers in $\tilde{B}$ arbitrarily to sellers such that everyone in $\tilde{B}$ transacts. Note that this matching does indeed match the top $|S|$ buyers to the top $|S|$ sellers, so it must be a maximum weight matching.

    By construction, since $b_i \in B_{vi}$ and $N_{G}(B_{vi})$ only transacts with the largest buyers in $B_{vi}$, it follows that the minimum opportunity path must point to a buyer in $B_{vi}$ and the price is thus bounded below by $\min_{b_i \in B_{vi}} v_i \geq k$.
\end{proof}

As a subroutine that will be used in our final algorithm, we show that we can efficiently decide, given a graph $G$, whether there exists a Hall violator containing a buyer $b$. To begin, we define the following graph-theoretic problems.

\begin{definition}[vertex Hall violator problem]
    Given a bipartite graph $G=(B\cup S,E)$ and a vertex $b\in B$, is there a Hall violator that contains $b$?
    \label{def:hall_violator_vertex}
\end{definition}
\begin{definition}[max-difference Hall violator problem]
    Given a bipartite graph $G=(B\cup S, E)$ and a positive integer $k$, is there a Hall violator $B_v\subset B$ such that $|B_v|-|N_G(B_v)|\geq k$?
    \label{def:max_dif_hall_violator}
\end{definition}

\begin{lemma}
\label{lemma:hv_reduction}
    There is a polynomial-time reduction from the vertex Hall violator problem to the max-difference Hall violator problem.
\end{lemma}
\begin{proof}
    We show that the vertex Hall violator problem can be reduced in polynomial time to the max-difference Hall violator problem. Given an instance $(G, b)$ of the vertex Hall violator problem, let $G'=(B\setminus\{b\}\cup S\setminus\{N_G(b)\},E)$ and $k=|N_G(b)|$. Then, solve the max-difference Hall violator problem on the instance $(G', k)$. If there is a yes-certificate $B_{v}\subseteq B\setminus \{b\}$ for the max-difference Hall violator problem, then $B_v\cup\{b\}$ is a Hall violator that includes $b$ in $G$. This is because by the yes-certificate
    \begin{eqnarray*}
        |B_v| - |N_{G'}(B_v)| & \geq & k = |N_G(b)|\\
        |B_v\cup\{b\}| - 1 & \geq & |N_{G'}(B_v)|+|N_G(b)| = |N_G(B_v\cup \{b\})|
    \end{eqnarray*}
    If $(G', k)$ is a no-instance of the max-difference Hall violator problem, then there does not exist a Hall violator that contains $b$ in $G$.  Otherwise, if $B_v\cup\{b\}, b\notin B_v$ were such a Hall violator, then
    \begin{eqnarray*}
        |B_v\cup\{b\}| -1 & \geq & |N_G(B_v\cup\{b\})| = |N_{G'}(B_v)|+|N_G(b)|\\
        |B_v| - |N_{G'}(B_v)| & \geq & |N_G(b)| = k
    \end{eqnarray*}
    which contradicts with $(G',k)$ being a no-instance.
\end{proof}

Finding the maximum-difference Hall violator is closely related to the notion of the deficiency of a graph, defined below.

\begin{definition}
\label{def:deficiency}
    Define the deficiency of a subset $B_v$ as $\mathrm{def}(B_v) = \max \left(|B_v| - |N_G(B_v)|, 0 \right)$. Similarly, define the deficiency of $B$ as $\mathrm{def}(G; B) = \max_{B_v \subseteq B} \mathrm{def}(B_v)$.
\end{definition}

For bipartite graphs, we have the following well-known fact about the deficiency of a graph $G$.

\begin{lemma}
\label{lem:deficiency}
    Let $m$ be the size of the maximum-matching on $G$. Then
    \begin{align*}
    \mathrm{def}(G; B) = |B| - m
    \end{align*}
\end{lemma}

Moreover, one can find the subset of maximum deficiency in polynomial time.

\begin{theorem}
    There is a polynomial-time algorithm to find the max-difference Hall violator (or equivalently, the subset of $B$ with maximum deficiency).
\end{theorem}\label{thm:poly_max_diff_violator}
\begin{proof}
    Note that finding the max-difference Hall violator is precisely equivalent to finding the subset of maximum deficiency. Start with a maximum matching $M$ on $G$, and let $B_{unmatched}$ be the subset of vertices in $B$ that are not saturated by $M$. If $B_{unmatched}=\emptyset$ then Hall theorem says there are no Hall violator. So we only consider $B_{unmatched}\neq \emptyset$. Let $B_{max}$ be the set of all vertices that are reachable via an alternating path from $B_{unmatched}$.

    We claim that the set $N_G(B_{max})$ is fully saturated by $M$. Consider any $s \in N_G(B_{max})$. By construction, $s$ is reachable from some unmatched vertex via an alternating path. If $s$ were not saturated by $M$, this would represent an augmenting path, which is a contradiction since $M$ is maximal.

    Additionally, we have that every $s \in N_G(B_{max})$ is matched to a vertex in $B_{max}$, again by construction of $B_{max}$ by the opportunity path. Thus, there is a bijection between matched vertices in $B_{max}$ and $N_G(B_{max})$, so by Lemma \ref{lem:deficiency}, it follows that $|B_{max}| - |N_G(B_{max})| = |B_{unmatched}| = \mathrm{def}(G; B)$, and $B_{max}$ is the max-difference Hall violator.
\end{proof}

As a corollary of Lemma~\ref{lemma:hv_reduction} and Lemma~\ref{thm:poly_max_diff_violator}, we can solve the vertex Hall violator problem in polynomial time.

\begin{corollary}
    The vertex Hall violator problem can be solved in polynomial time.
\end{corollary}
\begin{proof}
    By Lemma~\ref{lemma:hv_reduction}, the vertex Hall violator problem can be reduced to the max-difference hall violator problem, which can be solved in polynomial time using the algorithm described in Lemma~\ref{thm:poly_max_diff_violator}.
\end{proof}

Finally, we have all the tools necessary to conclude the proof of Theorem~\ref{thm:hom_min_nm_approx}. We restate the theorem, along with the proof below.

\begin{theorem}
    Given a buyer-seller pair $(b_i, s_j)$ not connected via a world edge, with $b_i \in \tilde{B}$, we can find the set of platform edges that maximizes the revenue generated from this pair in polynomial time.
\end{theorem}
\begin{proof}
    By Lemma~\ref{lemma:hall_violator_prices}, we simply need to find the Hall violator $B_{vi} \subseteq \tilde{B}$ containing $b_i$ of $G^{-j} = (S \setminus \{s_j\} \cup \tilde{B}, E)$ that maximizes the minimum valuation among all buyers in $B_{vi}$.

    Note that $s_j$ sells to $b_i$ at one of $|\tilde{B}|$ prices. To find the max revenue, we can simply check these candidate prices by iteratively checking if there is a Hall violator $B_{vi}$ for $G^{-j} = (S \setminus \{s_j\} \cup \tilde{B}, E)$ such that $b_i \in B_{vi}$, removing the lowest-value bidders from the graph at each round.

    That is, we look for a Hall violator containing $b_i$, allowing all buyers in $\tilde{B}$ to be used. We repeat this process, allowing all buyers except for $b_{(m)}$ to be used. Next, we allow all buyers except for $b_{(m)}, b_{(m - 1)}$ to be used, continually removing the buyers with the smallest values until no such Hall violator exists. Note that this will indeed find us the Hall violator that includes $b_i$ with optimal maximin value.

    After finding the Hall violator with the optimal maximin value, the second statement of Lemma~\ref{lemma:hall_violator_prices} gives a way to construct platform edges to sell $s_j$ to $b_i$ at such price, concluding the proof.
\end{proof}

\HomoMinNM*
\begin{proof}
    The revenue optimal platform matching has at most $\min\{n,m\}$ edges that transact. One of the edges will generate more than $\frac{1}{\min\{n,m\}}$. By Lemma 13.8, the platform can iterate through all buyer-seller pairs $(b_i, s_j)$ and check for the maximum revenue attainable from this pair. The best pair guarantees a $\frac{1}{\min\{n, m\}}$-fraction of the optimal revenue.
\end{proof}
\end{document}